\documentclass[11pt]{article}

\usepackage{amsmath,amssymb,amsfonts}
\usepackage[ruled,vlined,linesnumbered,noend]{algorithm2e}
\usepackage{cite}
\usepackage{physics}
\usepackage{geometry}
\usepackage{microtype}
\usepackage{tikz}
\usepackage{amsthm}
\usepackage{tabularx}
\usepackage{array}
\usepackage{xcolor}
\usepackage{multicol}
\usepackage{graphicx}
\usepackage{hyperref}
\usepackage{times}
\usepackage{subcaption}
\usepackage{float}

\newcommand{\UNLABELED}{\textsc{Unlabeled}}
\newcommand{\INLAB}{\textsc{In}}
\newcommand{\OUTLAB}{\textsc{Out}}

\theoremstyle{plain}
\newtheorem{lemma}{Lemma}
\newtheorem{theorem}{Theorem}

\theoremstyle{definition}
\newtheorem{definition}{Definition}

\theoremstyle{remark}

\newtheorem{fact}{Fact}

\newtheorem{corollary}{Corollary}[section]

\newcommand{\out}{\textsc{OUT}}

\newcommand{\SBL}{\textsc{Small Blossom Liquidation}}
\newcommand{\FVR}{\textsc{Free Vertex Reduction}}
\newcommand{\Finalization}{\textsc{Finalization}}
\newcommand{\Liq}{\textsc{Liquidationist}}
\newcommand{\QSearchOne}{\textsc{QSearchOne}}
\newcommand{\SearchOne}{\textsc{SearchOne}}

\renewcommand{\emph}[1]{\textbf{\textit{#1}}}

\newcommand{\GE}{\textsc{Grow Event}}
\newcommand{\BE}{\textsc{Blossom Event}}

\SetInd{0.4em}{0.8em}
\SetNlSty{}{:}{}
\SetAlgoNlRelativeSize{-1}
\DontPrintSemicolon

\title{A Quantum Scaling Algorithm for Maximum-Weight Perfect Matching in General Graphs}

\author{ \begin{tabular}{c@{\hspace{4em}}c@{\hspace{4em}}c} Kourosh Mirsohi & Sandy Irani & Michael T. Goodrich \end{tabular}\\[0.9em] Department of Computer Science\\ University of California, Irvine }

\date{}

\begin{document}
\maketitle

\begin{abstract}
Quantum speed-ups have been obtained for many fundamental graph problems, including most variants of matching.
A notable exception, however, is the maximum-weight perfect matching
(MWPM) problem in general graphs with integer edge weights, which is arguably the most challenging variant of matching.
We present a quantum algorithm for MWPM in general graphs
that runs in \(
    \widetilde{O}(n m^{2/3}\log W)
\)
time, where $W$ is an upper bound on the magnitude of the edge weights.
This is an
improvement over the best known classical combinatorial bound of
\(
    \widetilde{O}(m\sqrt n \log W)
\)
in the dense regime, where $m\ge n^{3/2}$.
To the best of our knowledge, this is the first quantum algorithm to obtain an
asymptotic improvement over the best classical combinatorial algorithm for the MWPM problem in general
graphs. The running time of our method accounts for QRAM initialization and access
overheads up to polylogarithmic factors, as well as all classical updates to the 
data structures.
At a high level, our algorithm is based on a classical framework due to Duan, Pettie, and Su, but our algorithm requires replacing certain classical tasks with quantum methods, alternative analysis of classical procedures, and the use of alternative data structures that can be implemented effectively in the QRAM model.

\end{abstract}

\section{Introduction}

Given an undirected $n$-vertex graph, $G=(V,E)$, with integer weights on its $m$ edges, the \emph{maximum-weight perfect matching} (MWPM) problem is to find a set of edges, $M$, of
maximum weight with no common endpoints such that every vertex in $G$ is incident to
an edge in $M$.
The MWPM appears directly in applications such as resource allocation, exchange
markets, scheduling, and Christofides' well-known $3/2$-approximation for metric
TSP~\cite{Christofides76}. The MWPM problem and its variants have been extensively studied since the
pioneering work by Edmonds~\cite{Edmonds65poly,Edmonds65ptf} in 1965; 
see, e.g.,~\cite{Kuhn55,Munkres57,HopcroftKarp73,MicaliVazirani80,Gabow85scaling,
GabowTarjan91,Lawler76,LovaszPlummer86,Schrijver03}. 
Note that maximum-weight matching
without the perfect-matching requirement reduces to MWPM with only constant-factor
overhead~\cite{Pettie12}, so MWPM is a natural canonical formulation of the problem.

MWPM is arguably the most computationally demanding matching problem.
Variants of the problem restricted to bipartite graphs or unweighted graphs
are generally  easier. 
One of the difficulties encountered by matching problems on general graphs
is that searches for an augmenting path of alternating matched and unmatched edges, beginning and ending with an unmatched vertex, can encounter odd
cycles.
Edmonds'  algorithm handles
odd cycles by temporarily treating the corresponding odd set of vertices as a single
contracted object, called a \emph{blossom}. 
While blossoms are contracted and dissolved throughout the algorithm, 
the set of blossoms forms a nested, non-crossing family of odd sets,
often called a \emph{laminar} family. 
Thus, at any point in time in the execution of the algorithm, the current set of blossoms is organized as a forest of blossoms. 
Edmonds' algorithm also follows a primal-dual framework, which 
 maintains auxiliary
\emph{dual variables} that are used to determine whether a particular edge is tight 
and therefore eligible to be considered when searching the graph.
Dual variables are occasionally updated (a process called \emph{dual adjustment})  so that new useful 
edges become available while feasibility is preserved.

These techniques have  continued to be useful over time. 
For example,
the classical MWPM algorithm of Duan, Pettie, and
Su~\cite{dps}, which runs in $\widetilde{O}(m\sqrt{n}\log W)$ time, incorporates all of these ideas.
Their algorithm searches a \emph{contracted graph}, where all blossoms are contracted so that the roots of the blossom trees are treated as single vertices and only almost-tight edges between the blossoms are considered.

Meanwhile, quantum algorithms for graph problems have been studied extensively over the past two decades, yielding speedups in running time or query complexity for
fundamental problems including shortest paths and minimum spanning
trees~\cite{DHHM04,NayebiWilliams14}, graph connectivity~\cite{DHHM04,JarretEtAl18},
maximum flow~\cite{AS06}, graph sparsification, and cut approximation~\cite{ApersDeWolf20}, as well as  several variants of
matching~\cite{AS06,Dor09,LL15,BT20,KW21,Bli22,TM24}.

Quantum algorithms for matching, in particular, have been obtained for almost every variant of matching in bipartite and general graphs (see Table \ref{tab:quantum-matching-prior} and the discussion in Appendix~\ref{sec:qmatchreview}). 
A notable exception is MWPM in general graphs with integer edge weights.
In this work, we present a quantum algorithm for MWPM in general graphs
that runs in \(
    \widetilde{O}(n m^{2/3}\log W)
\)
time, where $W$ is an upper bound on the magnitude of the edge weights.
This is an
improvement over the best known classical combinatorial bound of
\(
    \widetilde{O}(m\sqrt n \log W)
\)
in the dense regime~\cite{GabowTarjan91,dps},\footnote{As is common, 
we use $\widetilde{O}(*)$ to ignore polylogarithmic factors in $n$ and $m$.} where $m\ge n^{3/2}$. 
To the best of our knowledge, this is the first quantum algorithm to obtain an
asymptotic improvement over the best classical algorithm for the MWPM problem in general
graphs.

\subsection{The Quantum Random Access Memory (QRAM) Model}
The model of computation we use in this paper is the 
\emph{quantum random access memory} (QRAM) model described by Giovannetti, Lloyd, and
Maccone~\cite{GLM08qram,GLM08arch}, and further utilized and studied in work on quantum algorithms~\cite{Prakash14,KP17,Arunachalam2020,Park19,Hann2021,Xu23qram,Phalak23qram}. 
 The QRAM model has a shared memory 
equipped with dual interfaces---one classical and one quantum. 
Formally, if $A[0],A[1],\dots,A[N-1]$
is such a shared memory array, 
we assume any cell, $A[i]$, can be accessed classically, given $i$, in one step;
hence, each entry can be read or written in $O(1)$ time.
Individual entries can be updated in $O(1)$ time.
We also assume access to a unitary ``access'' oracle,
$O_A:\ket{i}\ket{0}\longmapsto \ket{i}\ket{A[i]}$.  A quantum address register may be in superposition, so for
any state $\sum_i \alpha_i \ket{i}\ket{0}$, 
one application of $O_A$ yields $\sum_i \alpha_i \ket{i}\ket{A[i]}$, which 
we view as a coherent read of the shared memory.

This framework allows
updates to $A$ to be performed sequentially with classical 
control between quantum subroutine calls that can utilize quantum address registers in superposition. 
Importantly, our use of this framework avoids so-called ``cheats'' where quantum 
algorithms perform
computations in sublinear time while ignoring the time needed to load their inputs.
For example, in our use of the QRAM model, 
the time to load an input graph into memory is completely accounted for in our running times.

This framework also allows us to use a generalized ``BBHT'' variant of the well-known
Grover-style search quantum subroutine~\cite{Grover96,BBHT98} using the shared memory. 
For example, after having loaded
a graph into the QRAM shared memory, we can then perform a graph
traversal that iterates across vertices, such that, for each vertex, $v$, visited during
the traversal
we can choose an unvisited edge incident to $v$ in
$O(\sqrt{d_v/k})$ time, where $d_v$ is the degree of $v$ and $k$ is the number of unvisited
edges incident to $v$. Owing to the inherent difficulty of the MWPM
problem, however, our algorithm involves much more than just applying Grover-style searching
to improve the running time for various graph traversals.
Instead, as we discuss below, our algorithm involves significant modifications to classical
data structures and procedures  as well as retuning techniques from previous algorithms,
most notably the classical MWPM algorithm of Duan, Pettie, and Su~\cite{dps}.

\subsection{Main Result and Technical Contributions}
At a high level, our algorithm involves a non-trivial adaptation and modification of the
MWPM algorithm of Duan, Pettie, and Su~\cite{dps} to the QRAM model.
The authors call their method the ``{\Liq}'' algorithm,
which is organized around
three procedures, \SBL{}, \FVR{}, and \Finalization{}.
We obtain an asymptotically faster quantum algorithm by speeding up the
\FVR{} procedure. The other two procedures are completely classical and are the same as in the 
MWPM algorithm of Duan, Pettie, and Su~\cite{dps}. 

For integer edge weights, the \Liq{} algorithm uses \emph{weight scaling},
introduced by Gabow~\cite{Gabow85scaling}, where the edge weights are revealed
gradually over \(O(\log W)\) scales.
Within each scale, the algorithm performs \SBL{} and \FVR{}, followed by some additional steps to complete that scale. After the final scale,
\Finalization{} is applied to obtain an exact maximum-weight perfect matching.

One of the primary challenges with this approach is to maintain the
 matching, dual variables, and blossom structures from one scale to the next.
The \Liq{} algorithm balances the cost of the three procedures by partitioning blossoms into \emph{small} and \emph{large} blossoms based on a threshold parameter, $\tau$. 
The dominant contributions to the total running time are summarized in
Table~\ref{tab:procedure-costs}. The factor \(\log W\) accounts for the
\(O(\log W)\) scaling iterations. Although \Finalization{} is performed only
after the last scale, its input contains the free vertices accumulated across
all scales.

\begin{table}[hbt]
\centering
\renewcommand{\arraystretch}{1.2}
\begin{tabular}{|c|c|c|}
\hline
\textbf{Procedure}
&
\textbf{\Liq{}}
&
\textbf{This work}
\\
\hline
\SBL{}
&
\rule[-8pt]{0pt}{24pt} \(\widetilde{O}(\tau m\log W)\)
&
\(\widetilde{O}(n\tau^{2}\log W)\)
\\
\hline
\FVR{}
&
\rule[-8pt]{0pt}{24pt} \(\widetilde{O}(\tau m\log W)\)
&
\(\widetilde{O}(\tau\sqrt{mn}\log W)\)
\\
\hline
\Finalization{}
&
\(\widetilde{O}\!\left(\dfrac{mn}{\tau}\log W\right)\)
&
\rule[-11pt]{0pt}{28pt} \(\widetilde{O}\!\left(\dfrac{mn}{\tau}\log W\right)\)
\\
\hline
\end{tabular}
\caption{Running times of the main procedures in the classical \Liq{} algorithm and our work.}
\label{tab:procedure-costs}
\end{table}

The \FVR{} procedure iterates a subroutine called \textsc{SearchOne} $\tau$ times.
The classical running time of \textsc{SearchOne} is $O(m)$, and we improve the running time to $\widetilde{O}(\sqrt{mn})$ on a quantum computer.
As mentioned above,
the \SBL{} and \Finalization{} procedures are unchanged. For \SBL{}, however, we use a different analysis that gives the bound \(\widetilde{O}(n\tau^2\log W)\) for simple graphs. 
As shown by Duan, Pettie, and Su~\cite{dps}, for the \SBL{} procedure, one need only consider edges that are internal to small blossoms. Furthermore, the running time to liquidate each small blossom is the number of vertices in that blossom (upper bounded by $\tau$) times the number of edges internal to that blossom. Summing up over all the small blossoms gives a running time of \(\widetilde{O}(\tau m)\) for each scale.
We observe that the number of edges within a small blossom can be at most $\tau^2$ and 
then summing up over all the blossoms gives an upper bound of $O(n)$ on the total number of vertices.
Thus the running time of \SBL{} can alternatively be upper bounded by
\(\widetilde{O}(n\tau^{2})\) for each scale (see Lemma \ref{le:detail}).

For the original classical \Liq{} algorithm, the overall running time is minimized by setting $\tau = \sqrt{n}$, yielding an overall running time of $\widetilde{O}(m\sqrt{n}\log W)$.
The overall running time of our algorithm is minimized by setting $\tau = m^{1/3}$,
yielding a running time of $\widetilde{O}(n m^{2/3}\log W)$, improving on the classical  bound  in the dense regime, $m \ge n^{3/2}$.
This is indeed the main result of this paper and it
represents the first asymptotic improvement for a quantum algorithm over the best classical algorithm for MWPM in any regime, as
we are not aware of any prior quantum MWPM algorithm.

As discussed above,
our quantum algorithm is for the QRAM model, in which classical arrays can be accessed in quantum superposition. The running time of our algorithm accounts for all classical updates to the QRAM arrays, including the cost of loading the initial graph into the QRAM.
Note that since the overall running time of \FVR{} is $\Omega(m)$,
we do not require any persistent memory in QRAM from one iteration to the next.
In fact, since each call to \FVR{} happens in a different scale, the weights of all the edges must be updated. 

In order to illustrate the challenge in implementing each of the $\tau$ calls to \textsc{SearchOne} in 
$\widetilde{O}( \sqrt{mn})$ time, we consider  just the first task of \textsc{SearchOne},
which finds a maximal set of blossom-disjoint eligible augmenting paths from the free blossoms, using an ordered depth-first search with temporary search-blossom contractions. The search is performed in the contracted graph in which blossoms are contracted into vertices and only certain edges are eligible. 
The problem for the quantum algorithm is that we do not have time to build and maintain the contracted graph. Instead, we show how to search using the original graph. Once a vertex is reached, all of the vertices with which it shares a blossom are added to the search stack. 
The search algorithm still needs to  ignore edges contained within a blossom, of course, since only edges that span two blossoms should be considered. To this end, our algorithm maintains a vertex-indexed array (called $\textsc{RootB}(\cdot)$) indicating the root blossom to which it belongs. 
Our analysis establishes that the required information can be maintained efficiently and that the path in the contracted graph can be efficiently reconstructed from the information returned by the quantum algorithm that searches the original graph. 

Graph search is only the first of several tasks performed in \QSearchOne{}, our
quantum version of \textsc{SearchOne}. For example,
\QSearchOne{} augments along a maximal set of eligible vertex-disjoint
augmenting paths, contracts the eligible odd cycles
 encountered by the search into blossoms, performs dual
adjustment, and dissolves blossoms that reach a certain condition.
We prove that our \QSearchOne{} performs the same tasks 
as \SearchOne{} while running in $\widetilde{O}(\sqrt{mn})$ time.

\section{Our Quantum \FVR{} Procedure}
\label{sec:result}

\subsection{Blossom Basics}

Matching algorithms in general graphs maintain a collection of nested sets of vertices called \emph{blossoms}. A \emph{trivial} blossom is just a vertex from the original graph and a \emph{non-trivial} blossom contains more than one vertex. 
The nesting structure of the  blossoms is represented as a forest, denoted by $\Omega$,
where leaves are trivial blossoms and internal tree nodes are non-trivial blossoms. 
We say that a vertex $v$ from the original graph is \emph{contained} in blossom $B$, if $B$ is an ancestor of $v$ in the forest.
A \emph{root} blossom is the root of a tree in the forest. Note that some trees may consist of a single trivial blossom.

Each non-trivial blossom has an odd number of children, which are themselves blossoms. These blossoms  form an odd-length cycle.  If there are  $d$ children, then there are $\lfloor d/2 \rfloor$ matched edges along the cycle and one unmatched child, called the \emph{base} of the blossom.  
There can be at most one matched edge that spans the boundary of the blossom.
If such an edge exists, the blossom is said to be \emph{matched}, otherwise it is \emph{free}.
Figure \ref{fig:blossomtree} shows a portion of the original graph, the blossom structure, and the corresponding blossom tree.

\begin{figure}[htb]
\centering

\begin{minipage}{0.40\textwidth}
    \centering
    \includegraphics[width=.7\linewidth]{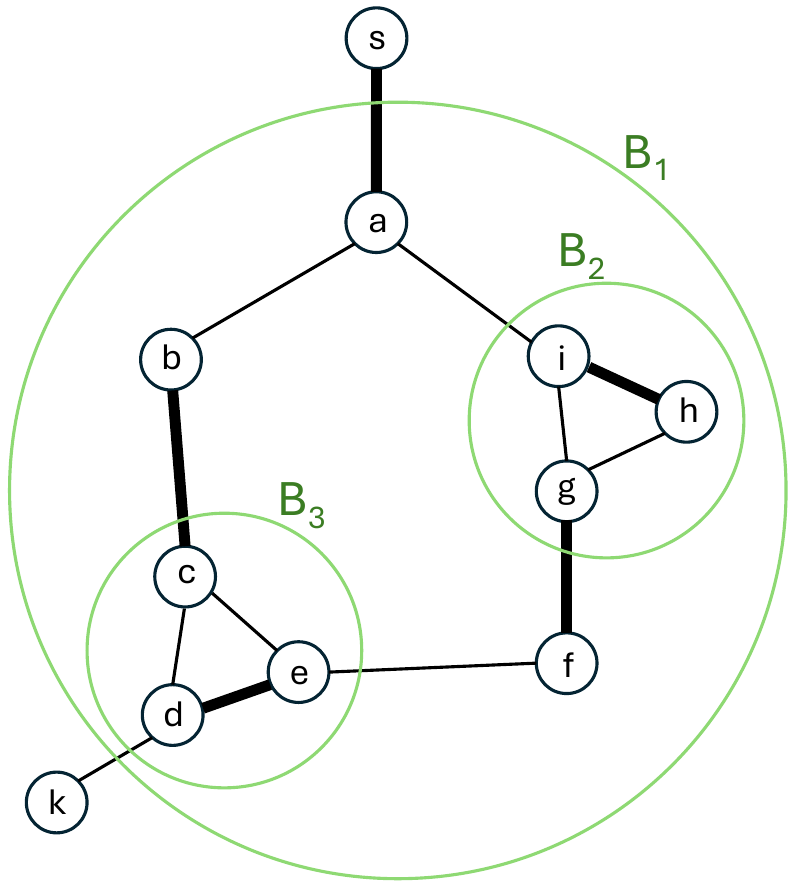}
    \vspace*{-8pt}
    \caption*{(a)}
\end{minipage}
\hspace{.3in}
\begin{minipage}{0.40\textwidth}
    \centering
    \includegraphics[width=\linewidth]{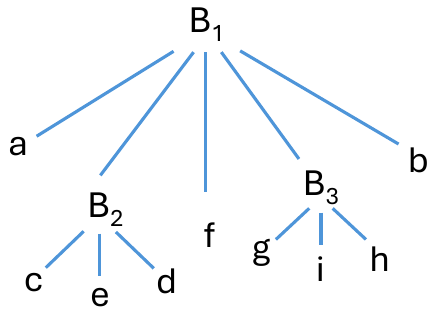}
    \caption*{(b)}
\end{minipage}
\vspace*{-6pt}
\caption{(a) An example of the blossom structure in a graph and (b) its corresponding blossom tree. Matched edges are shown in bold. All blossoms in the figure are matched. If the edge $(a,s)$ were not present in the matching, then $a$ would be a free vertex and $B_1$ would be a free blossom.}
\label{fig:blossomtree}
\end{figure}

The \Liq{} algorithm conducts searches for augmenting paths in
the contracted graph, denoted by $G/\Omega$, where
each  root blossom is contracted to a single vertex~\cite{dps}. 
Note that since the contracted graph may have multi-edges, a path in the contracted graph must specify the sequence of blossoms along the path in addition to the edges from the original graph that connect the blossoms. A path in the contracted graph is an augmenting path if the first and last blossoms are free and the edges connecting the blossoms alternate between matched and unmatched edges. Figure \ref{fig:blossompath} shows an example of an augmenting path in the contracted graph. The formal definition is given below. 

\begin{figure*}[hbt!]
\begin{center}
	\includegraphics[width=0.5\textwidth]{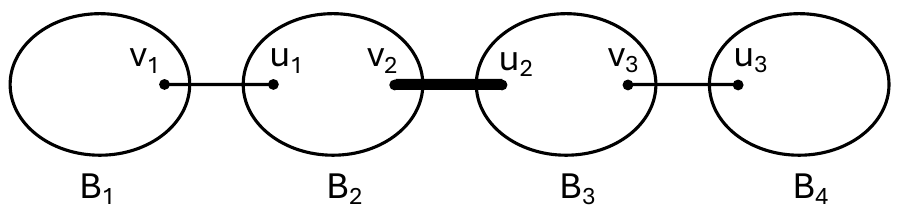}
        \end{center}
        \vspace*{-6pt}
        \caption{An example of an augmenting path in the contracted graph. The sequence of blossoms is $\langle B_1, \cdots B_r \rangle$.
    The endpoints of the edges $(u_i,v_i)$ are vertices in the original graph. The matched edge is bold.
    }
\label{fig:blossompath}

\end{figure*}

\begin{definition}
    An \emph{augmenting path of blossoms} in the contracted graph is a sequence of root blossoms
    $\langle B_1, \cdots B_r \rangle$ and a set of edges from the original graph $\langle e_1, \ldots e_{r-1} \rangle$.  Each edge $e_i$ has one endpoint in $B_i$ and one in $B_{i+1}$.
    The sequence of edges starts and ends with an unmatched edge, and alternates between matched and unmatched.
    Moreover, $B_1$ and $B_r$ are both free blossoms. 
\end{definition}

We will use the following facts about the blossom structure. It is standard for matching algorithms in general graphs to explicitly maintain the blossom structure so that these facts always hold.
For completeness, the facts below are proved in Appendix \ref{sec:contracted}.

\begin{fact}
\label{fact:aug1}
Any augmenting path of blossoms in the contracted graph can be lifted to an augmenting path in the original graph. The lifted path can be computed in time that is linear in the length of the lifted path. If the augmenting path in the contracted graph starts at blossom $B$ and ends at $B'$, then the lifted augmenting path starts at the base of $B$ and ends at the base of $B'$.
Augmenting along such a path does not change the invariants of the blossom structure. 
\end{fact}

\begin{fact}
\label{fact:aug2}
    If a set of  blossom-disjoint augmenting paths in the contracted graph is  lifted to the original graph, the resulting lifted paths are vertex-disjoint in the original graph. 
\end{fact}

\begin{fact}
    \label{fact:blossommatch}
    Each blossom contains at most one unmatched vertex.  If the blossom  contains an unmatched vertex, then the unmatched vertex is the base of the blossom and the blossom is free. If a blossom does not contain an unmatched vertex, then the blossom is matched.
\end{fact}

\subsection{The Primal-Dual Framework and Eligible Edges}

Like many algorithms for matching in general graphs, the  \Liq{} algorithm is  a primal-dual algorithm that maintains a dual variable
for each vertex and blossom~\cite{dps}. The dual variable for vertex $v$ is denoted as $y(v)$ and the dual variable for each blossom $B$ is denoted as $z(B)$. For more details on the primal-dual framework, see Appendix~\ref{ap:lp}. 

The edges under active consideration are called \emph{eligible}. 
Whether an edge $(u,v)$ is eligible or not depends on whether the edge is in the current matching as well
as it's  weight $w(u,v)$, as well as the values of $y(u)$, $y(v)$, and $z(B)$, for any blossom $B$ that contains both $u$ and $v$. 
One of the factors used to determine eligibility is whether an edge is \emph{tight} as determined by the values of the dual variables. 
The tightness condition used by Duan, Pettie, and Su~\cite{dps} depends on the following quantity:
\[
yz(u,v) = y(u)+y(v)+\sum_{B\in \Omega: e \subseteq B} z(B).\]
More specifically, an edge is \emph{tight} if
$w(u,v) = yz(u,v)$ and is \emph{almost tight} if $w(u,v)-2 = yz(u,v)$.
Note that it may not be possible to calculate the term \(\sum_{B\in \Omega: e \subseteq B} z(B)\) in constant time. However, the \FVR{} procedure only needs to determine whether an
edge $(u,v)$ is eligible when $u$ and $v$ are contained in different root blossoms.
In this case, the \(\sum_{B\in \Omega: e \subseteq B} z(B)\) term is equal to $0$
and can be  ignored. In the description of our algorithm, we will therefore use the quantity $y(u)+y(v)$
in lieu of $yz(u,v)$.

The eligibility condition used by  Duan, Pettie, and Su~\cite{dps} in their \FVR{} Procedure is called \emph{Criterion 2}.
The condition depends on whether the edge is matched. Since we will be primarily focused on edges which span two root blossoms, we define the Criterion 2 below restricted only to those edges:

\medskip
\noindent
{\bf Criterion 2:} If an edge $(u,v)$ has endpoints in two different root blossoms, the edge is \emph{eligible} if one of the following holds:
\begin{enumerate}
\setlength{\itemsep}{0pt}
    \item $(u,v)$ is matched and tight, i.e., $y(u) + y(v) = w(u,v)$.
    \item $(u,v)$ is not matched and almost tight, i.e., $y(u) + y(v) = w(u,v) - 2$.
\end{enumerate}

\subsection{\FVR{} Procedure Overview}

The outer loop of the \Liq{} algorithm runs $O(\log (nW))$ times~\cite{dps}.
In each iteration (called a \emph{scale}), one more bit of the weights is revealed.
The \FVR{} procedure is  executed exactly once
per iteration. Since \FVR{} is the only part of our algorithm with a 
quantum component, we describe it in detail here. We provide just enough background on the \Liq{} algorithm to describe the quantum portion of the procedure and prove that our algorithm completes the same set of tasks as its classical counterpart.
For completeness, Section \ref{sec:full} describes the rest of the \Liq{} algorithm and the analysis for the overall running time of our quantum algorithm.

At the beginning of the \FVR{} procedure, the adjacency list structure for the  graph accessed by the QRAM
is updated to reflect the new weights for that scale. 
\FVR{} consists of $\tau$ iterations of \SearchOne{}, where $\tau$ is a parameter that is optimized later.
Figure \ref{fig:s1overview} describes the tasks performed by \SearchOne{}. The main objective in this section is to establish that the same tasks can be completed by the quantum algorithm in time
$\widetilde{O}(\sqrt{mn})$ instead of time $\widetilde{O}(m)$ required by the classical version. 

\begin{figure}[hbt!]
\setlength{\fboxsep}{15pt}
\noindent\fbox{\begin{minipage}{.9\textwidth}
{\bf \underline{\QSearchOne{}/\SearchOne{} Tasks:}}
\begin{enumerate}
    \item {\bf Find Augmenting Paths:} Find a maximal set of blossom-disjoint eligible augmenting paths from the set of free blossoms $F$ in the contracted graph $G/\Omega$, where each root blossom is contracted into a single vertex.
    \item {\bf Augment Paths:} Lift the augmenting paths found in the first task to augmenting paths in the original graph $G$ and augment along those paths.  The matching is updated, but the blossom structure remains unchanged.
    \item {\bf Detect New Blossoms:} Build a new search forest of alternating paths from the set of free blossoms in the contracted graph $G/\Omega$. A blossom is labeled \textsc{OUT} if its distance from the root is even. A blossom $B$ is labeled \textsc{OUT} if there is an even length alternating path from a free blossom to $B$ in the contracted graph. Thus, an edge between two \textsc{OUT} blossoms in the same tree corresponds to an odd cycle. Iteratively search for \textsc{OUT}--\textsc{OUT} edges between two root blossoms in the same search tree and contract the corresponding odd cycle into a new blossom. Update, the search graph to reflect the new \textsc{OUT}/\textsc{IN} status of any blossom that  changed as a result of the merge. Repeat until no eligible \textsc{OUT}--\textsc{OUT} edges remain. 
    \item {\bf Adjust Dual Variables:} Update the dual variables: $y(v)$ is incremented if $v$ is \textsc{IN} and decremented if $v$ is \textsc{OUT}. $z(B)$ is incremented by $2$ if $B$ is \textsc{OUT} and decremented by $2$ if $B$ is \textsc{IN}. 
    \item {\bf Dissolve Zero-valued Blossoms:} For any root blossom $B$ such that $z(B) = 0$, dissolve $B$. 
\end{enumerate}  
\end{minipage}
}
\caption{Overview of the tasks in \QSearchOne{}/\SearchOne{}. These tasks are the same in
both the classical \SearchOne{} algorithm and our \QSearchOne{} algorithm; we use
this framework to show that our quantum algorithm correctly performs all the tasks required
in the classical \SearchOne{} method.
}
\label{fig:s1overview}
\end{figure}

In our algorithm, Tasks $1$ and $3$ will be replaced by quantum routines and the rest of the tasks remain classical. Note that even the  classical part of the procedure may need to update data stored in the QRAM for future iterations. 
The table in Figure \ref{fig:data}  shows the data structures used and maintained by \SearchOne{}. The third column indicates whether the data structure is accessed by QRAM, while the fourth indicates  whether it is local to one iteration of \QSearchOne{} or persistent throughout the algorithm.

\begin{figure}

\begin{center}
\begin{tabular}{|p{0.14\linewidth}|p{0.39\linewidth}|p{0.17\linewidth}|p{0.18\linewidth}|}
\hline
{\bf Name} & {\bf Function}  & {\bf QRAM access} & {\bf Persistence}\\
\hline
$adj(v,i)$ & Ordered pair $(a_v(i), w(v, a_v(i))$, where $a_v(i)$ is the $i^{th}$ neighbor of $v$ in the adjacency list representation and $i \in \{1, \ldots, d_v\}$. & Indexed by $(v,i)$ & Persistent and Static\\
\hline
$M$ & Current matching. $M[v]$ is the vertex to which $v$ is matched. If $v$ is a free vertex, then $M[v] = \perp$ & Not accessed & Persistent\\
\hline
$\Omega$ & Forest of nested blossoms & Not accessed & Persistent\\
\hline
$y$ & Dual variables for each vertex. & Indexed by $v$ & Persistent\\
\hline
$z$ & Dual variables for each blossom. & Not accessed & Persistent\\
\hline
$\textsc{RootB}[v]$ & The name of the root blossom to which $v$ belongs & Indexed by $v$ &  \QSearchOne{} \mbox{Local}\\
\hline
$UIO[v]$ & The current search label & Indexed by $v$ & \QSearchOne{} \mbox{Local}\\
\hline
$\mathrm{OutTime}[B]$ & Time reached in search & Indexed by $B$ & \QSearchOne{} \mbox{Local}\\
\hline
$\mathrm{predEdge}[B]$ & Original graph edge joining $B$ to parent in search structure & Not accessed & \QSearchOne{} \mbox{Local}\\
\hline
\end{tabular}
\end{center}
\vspace*{-11pt}
\caption{The data structures used and maintained by \SearchOne{}. The third column indicates whether the data structure is accessed by QRAM, and the fourth column indicates whether it is local to one iteration of \QSearchOne{} or persistent throughout the algorithm.}
\label{fig:data}
\end{figure}

The classical \Liq{} algorithm maintains most of the data structures in the table, as well as explicitly maintaining $G/\Omega$.
Our quantum version of the algorithm cannot afford the time to build and maintain the contracted graph $G/\Omega$, however, so we introduce
$\textsc{RootB}[\cdot]$ to quickly determine whether an edge spans two root blossoms.
In particular, $u$ and $v$ are contained in different root blossoms if and only if
$\textsc{RootB}[u] \neq \textsc{RootB}[v]$.
The $\textsc{RootB}[\cdot]$ array is rebuilt in $O(n)$ time at the beginning of
each iteration of \SearchOne{} and must be kept updated throughout each iterations.

\paragraph{Overview of Task 1 of \textsc{QSearchOne}: Quantum Path Search.}
\label{sec:QPQSearch_over}
The first task of \QSearchOne{}
 effectively searches the contracted graph for a maximal set
of blossom-disjoint augmenting paths from $F$, the set of free blossoms in $G/\Omega$.
The algorithm will then augment along these paths lifted to the original graph.

We call our quantum algorithm that performs this task \textsc{QuantumPathSearch}{}, and the pseudocode is given in Algorithm \ref{alg:qpathsearch}.
The algorithm is an adaptation of the DFS  augmenting-path search of
Gabow and Tarjan, described in Section $8$ of ~\cite{GabowTarjan91}. 
The Gabow-Tarjan search algorithm uses DFS to construct a forest of alternating search trees in the contracted graph, starting with a free blossom. When an augmenting blossom-path is found or the current search is exhausted, the algorithm starts again from a new, unlabeled free blossom.

Periodically an odd cycle is encountered that needs to be contracted into a new blossom.
The Gabow-Tarjan algorithm stores these newly formed blossoms implicitly in a search structure.
We store them explicitly as temporary blossoms used in the search.
This is accomplished by making a copy $\Omega'$ of the current blossom structure
$\Omega$ at the beginning of \textsc{QuantumPathSearch}{}.
$\Omega'$ is augmented with new blossoms throughout the search and then discarded 
at the end, restoring the original blossom structure $\Omega$
from the beginning of the search.
The temporary blossom structure enables us to use a  union-find labeling scheme to keep the
$\textsc{RootB}[\cdot]$ array up to date. 
The Gabow-Tarjan algorithm also tracks the time that any vertex contained in a blossom was first placed on the DFS search stack, corresponding to the time it was first labeled \out{}.
This value (which we call $\mathrm{OutTime}[B]$) is used to determine when a new blossom will be formed. 

Since a classical  algorithm  has the luxury of $O(m)$ time to maintain its data structures, it can build the contracted graph explicitly and search directly in that graph. In this scenario,
each root blossom of the initially contracted graph $G/\Omega$ is assigned
a label, initialized to \textsc{UNLABELED}. During the alternating DFS,
a newly reached blossom is labeled \textsc{IN} or
\textsc{OUT} according to the parity of its alternating path from the
free root. 
Since the paths in the DFS tree must alternate, edges from an \textsc{OUT} blossom to its \textsc{IN} children are unmatched.  Each \textsc{IN} blossom has at most one child in the DFS tree, to which it is connected by a matched edge.
When a new \textsc{OUT} blossom is reached, it is put on the stack so that its unmatched outgoing edges can be explored. 

In \textsc{QuantumPathSearch}{},  the search is carried out in the original graph $G$. Thus, the $UIO$ array stores the labels indexed by vertex.
When the first vertex in a root blossom is reached, all of the vertices in that blossom are given the same label. 
If the label is \textsc{OUT}, then all the vertices from that blossom are added to the DFS stack.
Since the classical portion of \textsc{QuantumPathSearch}{} has access to the entire blossom structure (denoted by $\Omega'$), it can enumerate the vertices inside the blossom. For blossom $B$, we use $V(B)$ to denote all the graph vertices contained inside $B$ (corresponding to the leaves of $B$ in $\Omega$). 
Thus,  the vertices in $V(B)$ can be enumerated in time that is linear in $|V(B)|$.
Note that $V(\textsc{RootB}[x])$ denotes the set of vertices contained in the same root blossom as vertex $x$.

When an \textsc{OUT} vertex $v$ is at the top of the stack, \textsc{QuantumPathSearch}{} has to find an eligible unmatched edge from $v$ that leads to a vertex outside its root blossom. This is the heart of the quantum portion of the algorithm: $\textsc{GroverSearch}(v)$ searches for the an actionable step in $v$'s adjacency list. 
The criteria for a candidate neighbor $u$
are given in Figure \ref{fig:grover2}.
An unsuccessful Grover Search returns $\perp$.

\begin{figure}[h]
\begin{center}
\noindent\fbox{\begin{minipage}{.95\textwidth}
\vspace{.1in}
$\textsc{GroverSearch}(v)$  searches the adjacency list of $v$ for a neighbor
$u$  satisfying the following eligibility criteria:
\begin{enumerate}
    \item
$y(v) + y(u) = w(u,v) - 2$ ~~~~~~~~~~~~~ (i.e., is the edge almost tight?) 

\item 
$\GE{}$: ~~~~~~$UIO[u] = \textsc{UNLABELED}$

{\em OR}

\BE{}: ~~~~~~$UIO[u] = \out{} ~~\wedge~~ \mathrm{OutTime}[\textsc{RootB}[u]] > \mathrm{OutTime}[\textsc{RootB}[v]].$
 
\end{enumerate}

\vspace{.1in}
\end{minipage}
}
\caption{Criteria for $\textsc{GroverSearch}(v)$}
\label{fig:grover2}
\end{center}
\end{figure}

We use a function $f_v(i)$ for the eligibility criterion for a neighbor of $v$ in our Grover
search procedure: $f_v(i)$ is equal to $1$ if and only if  the two conditions on the list above are satisfied.
Note that $f_v(i)$ can be computed in constant time, given the QRAM accessible arrays described in Figure \ref{fig:data}.
Notice that it is unnecessary to explicitly check that $u$ and $v$ are in different root blossoms ($\textsc{RootB}[v] \neq \textsc{RootB}[u]$), since 
$\textsc{GroverSearch}(v)$ is only performed if $UIO[v] = \out{}$.
In this case, the conditions for both a $\GE$ or a $\BE$ imply that
$u$ and $v$ are contained in different root blossoms.

The alternating search tree is stored by the
$\mathrm{predEdge}[\cdot]$ array, indexed by  root blossoms.
$\mathrm{predEdge}[B]$ is the  edge from the original graph
joining $B$ to its parent in the tree.
If $\mathrm{predEdge}[B] = (u,v)$, then $\textsc{RootB}[v] = B$
and the parent of $B$ in the search tree is $\textsc{RootB}[u]$.
The pseudocode for \textsc{QuantumPathSearch} is shown in Algorithm \ref{alg:qpathsearch}
in Section~\ref{sec:proofs}. We will illustrate the outcome of a 
$\GE$ separately from a $\BE$.

\begin{figure}[H]

\centering

\includegraphics[width=.7\linewidth]{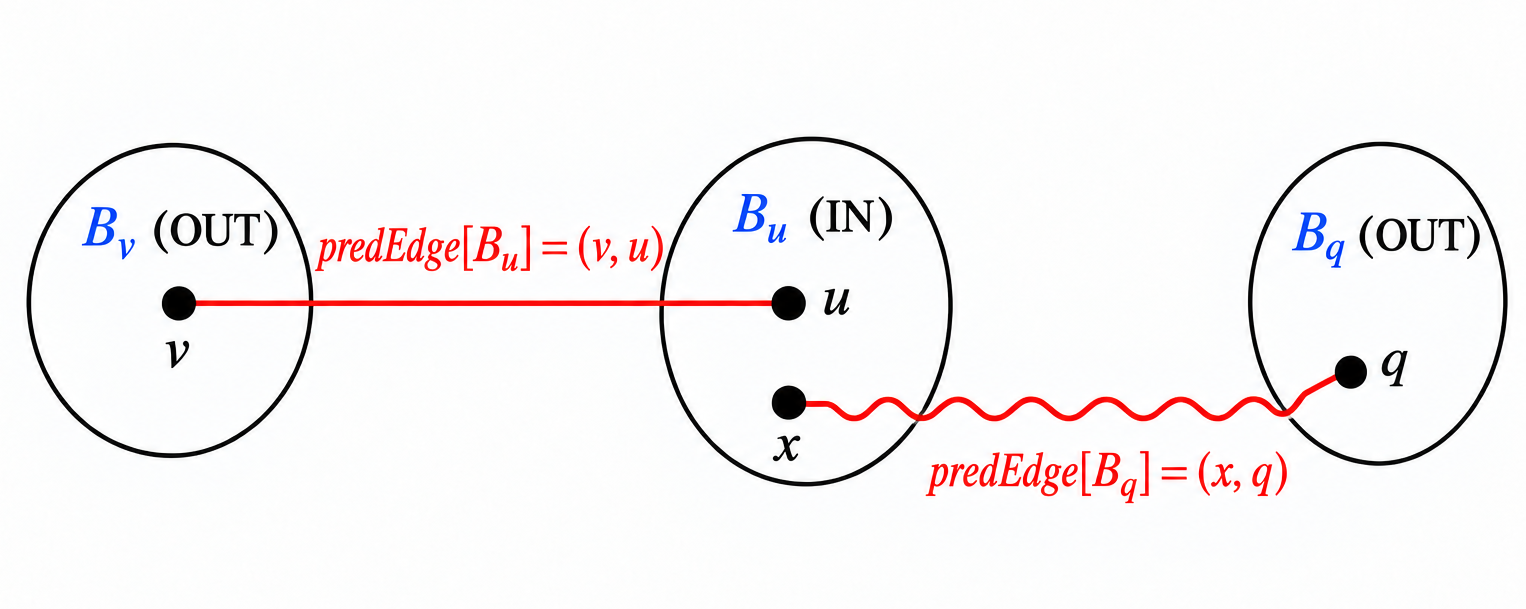}
\vspace*{-6pt}

\caption{A $\GE$ in \textsc{QuantumPathSearch}. The eligible unmatched edge $(v,u)$ reaches the previously
UNLABELED blossom $B_u$, so
$\mathrm{predEdge}[B_u]=(v,u)$ and $B_u$ becomes \textsc{IN}.
Since $B_u$ is matched, its eligible external matched edge
$(x,q)$, shown as a wiggly edge, is followed to the
\textsc{OUT} blossom $B_q$, with
$\mathrm{predEdge}[B_q]=(x,q)$.}

\label{fig:pathsearch}
\end{figure}

{\bf A $\GE$:}
Figure \ref{fig:pathsearch} shows the outcome of a single $\GE$.
A $\GE$ happens if the result of \textsc{GroverSearch}
applied to the vertex $v$ at the top of the stack
finds eligible edge $(v,u)$, where $u$ is unlabeled.
If $B_u$ denotes $u$'s root blossom, then 
 $\mathrm{predEdge}[B_u]$ array records the edge $(v,u)$
 in order to trace the blossom path up the search tree.
All of the vertices in $u$'s root blossom, denoted by $V(\textsc{RootB}[u])$, are labelled \textsc{IN}.

The next task is to determine whether blossom $\textsc{RootB}[u]$ is matched.
Figure \ref{fig:pathsearch} illustrates the case where $\textsc{RootB}[u]$ is matched and its outgoing matched edge $(x,q)$ is eligible. 
Vertex $x$ is contained in $\textsc{RootB}[u]$ and 
$q$ is contained in a new blossom $B_q = \textsc{RootB}[q]$.
Then
 $\mathrm{predEdge}[B_q]$ array records the edge $(x,q)$
 in order to trace the blossom path up the search tree.
 All of the vertices in $q$'s root blossom, denoted by $V(\textsc{RootB}[q])$, are labelled \textsc{OUT} and placed on the search stack to continue the DFS search.
Note that it is possible that blossom $\textsc{RootB}[u]$ is matched but its outgoing matched edge is ineligible. In this case, the search backtracks to vertex $v$.

Another outcome of a $\GE$ is that $\textsc{RootB}[u]$ is unmatched.
Note that because of Fact \ref{fact:blossommatch}, there can be at most one unmatched vertex in $\textsc{RootB}[u]$. If an unmatched vertex $x$ in $\textsc{RootB}[u]$ is discovered, that signals that an augmenting path has been found. Following the edges in the $\mathrm{predEdge}$ array
back to the free root
gives the current  alternating path; temporary blossoms on that
path are expanded from their defining cycles in $\Omega'$. The resulting
augmenting path in the original contracted graph is added to $\Psi$, while the matching itself is not changed
until Task~2.

{\bf A $\BE$:}
There is a second possible outcome of
\textsc{GroverSearch}$(v)$. If $u$ is already labeled
\textsc{OUT}, then the edge $(v,u)$ may close an odd alternating
cycle. Following the ordered DFS of Gabow and Tarjan, such
an edge is acted upon only when
\[
\textsc{OutTime}[\textsc{RootB}[u]]
>
\textsc{OutTime}[\textsc{RootB}[v]].
\]
The Gabow-Tarjan proof (and our Corollary \ref{cor:blossomevent})
shows that in this case, $\textsc{RootB}[u]$ is always a descendant
of $\textsc{RootB}[v]$ in the search tree, as depicted in Figure \ref{fig:blossomdetect}.
The corresponding odd cycle is then contracted into a temporary
search blossom in $\Omega'$, while the persistent blossom family
$\Omega$ is unchanged. The IN blossoms in the newly formed blossom
become OUT. All of the graph vertices in those  newly labeled OUT blossoms and are inserted into the DFS stack, in  descending order 
along the search tree path, as required
by the Gabow--Tarjan DFS. The new temporary blossom inherits the
 $\mathrm{predEdge}$ value of  base, and its
defining cycle is stored in $\Omega'$.

\begin{figure}[htb]
\centering

\begin{minipage}{0.40\textwidth}
    \centering
    \includegraphics[width=.76\linewidth]{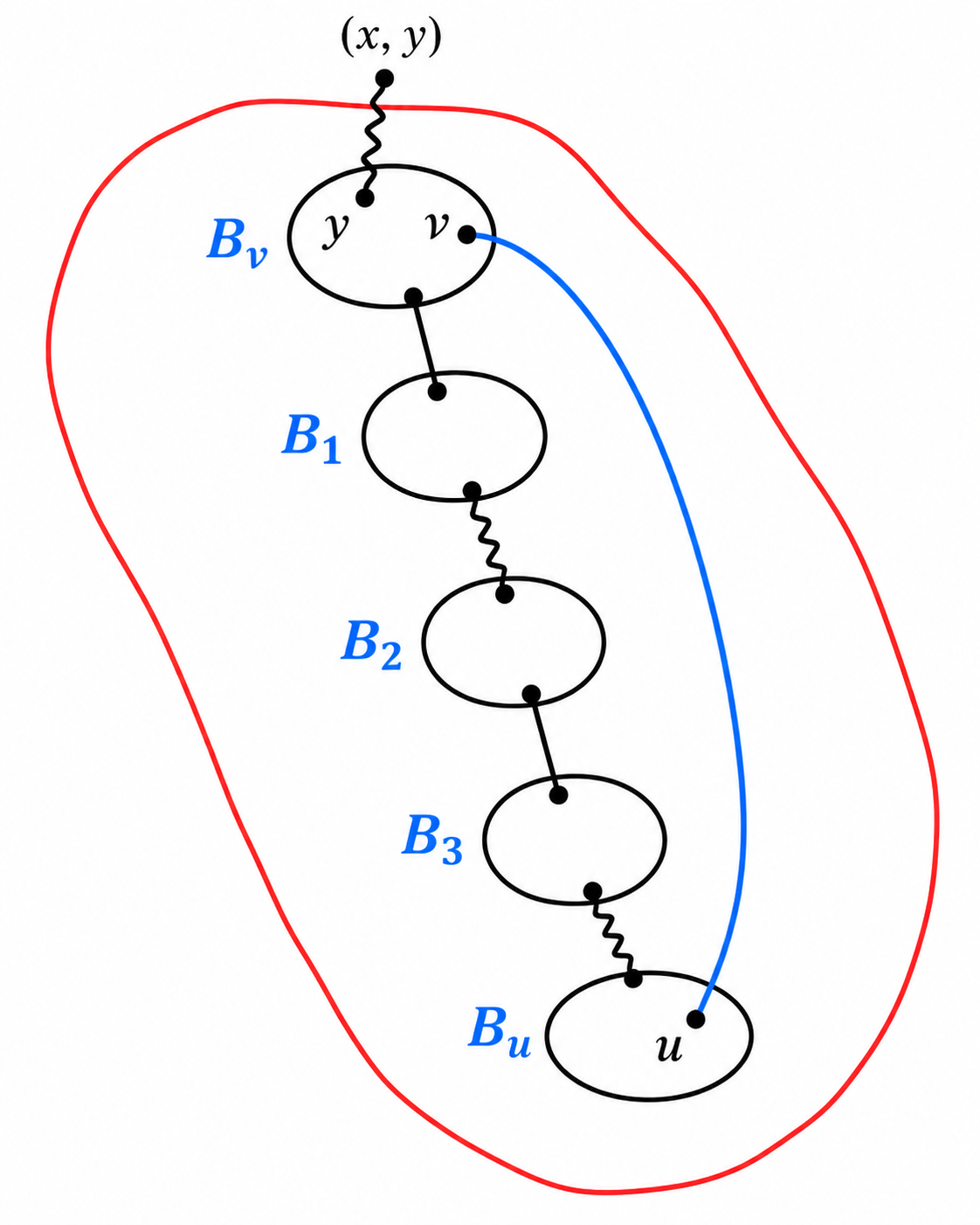}
    \vspace*{-8pt}
    \caption*{(a)}
\end{minipage}
\hspace{.3in}
\begin{minipage}{0.40\textwidth}
    \centering
    \includegraphics[width=\linewidth]{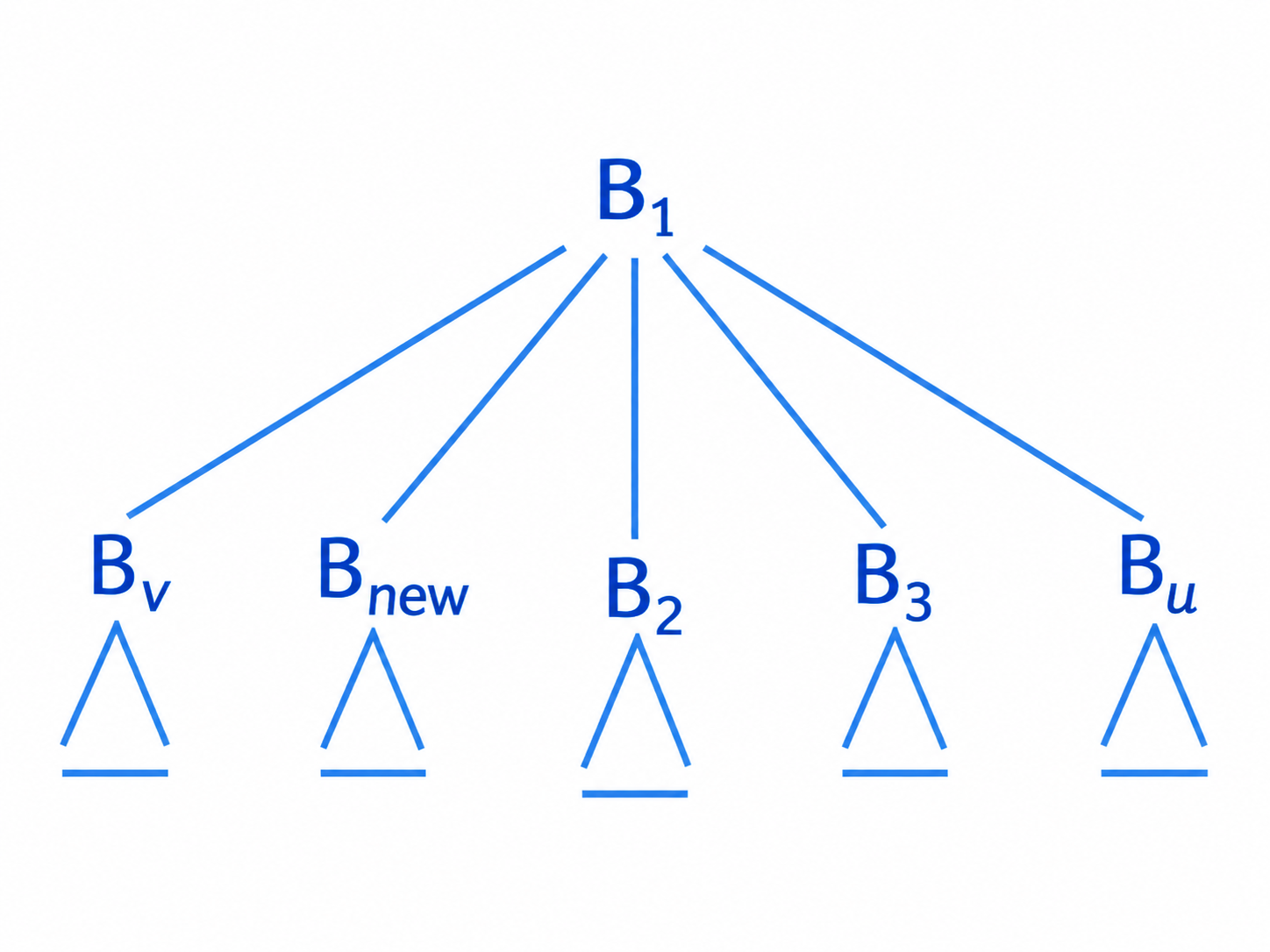}
    \caption*{(b)}
\end{minipage}
\vspace*{-6pt}
\caption{(a) The figure on the left shows $\textsc{OUT} - \textsc{OUT}$ edge $(u,v)$ and the odd cycle in the DFS tree that will be contracted. Assume wlog that $B_1$ is the blossom with the most vertices in the cycle. $B_1$ will receive a new label $B_{new}$ and the new blossom will get label $B_1$. The $\textsc{RootB}[\cdot]$ value  must be updated for the vertices in $B_v, B_2, B_3, B_u$, but not $B_1$.
(b) The blossom forest $\Omega'$ shown on the right has a new blossom whose children are the blossoms in the cycle. 
The base of the new cycle is $B_v$. 
  $\mathrm{predEdge}$  for the new blossom will be $(x,y)$.}
\label{fig:blossomdetect}
\end{figure}

The classical blossom contraction task usually uses a union-find structure with path
compression to update the blossom structure. When a new blossom is created, we need to update the $\textsc{RootB}$ array so that all of the vertices in the new blossom have the same value. By updating the vertices in all the children of the new blossom, except the largest, we are implicitly implementing a union-by-size data structure based on labeling each item with the name of the set to which
it belongs. 
The reason we use this implementation is that following pointers in a classical tree-based union-find data structure cannot be incorporated into a Grover-style quantum search. 
As a result, we can upper bound the number of times
$\textsc{RootB}$ is updated for a particular vertex $x$ by $\log n$.
Informally, every time $\textsc{RootB}[v]$ is relabeled, the number of vertices in its root blossom at least doubles.
This is formally established in Lemma \ref{le:union-by-size}.
Figure \ref{fig:blossomdetect} illustrates an odd cycle found as the result of an eligible $\textsc{OUT} - \textsc{OUT}$ edge $(u,v)$ and the resulting changes to the blossom structure.

{\bf Clean Up:}
After all free roots have been processed, the search-local
forest $\Omega'$ and the corresponding temporary modifications to
$\textsc{RootB}$ are discarded. Thus Task~1 returns the
augmenting-path family $\Psi$ while leaving the persistent blossom
family $\Omega$ unchanged, as required by the augmentation step of
\textsc{SearchOne}.
The paths in $\Psi$ form a maximal set of blossom-disjoint augmenting paths,
where the blossoms are from the original contracted graph $G/\Omega$.

The correctness for \textsc{QuantumPathSearch} is proven in Section \ref{sec:qpathcorrect}
Lemma \ref{lem:qpath-valid} establishes that the paths returned are eligible and blossom-disjoint, and Lemma \ref{lem:qpath-maximal} establishes maximality.
For both proofs, the blossoms are root blossoms in the original blossom structure $\Omega$.
Lemma \ref{le:qpath-runtime}, proven below, establishes that 
\textsc{QuantumPathSearch} runs in time
$\widetilde{O}(\sqrt{mn})$. The analysis uses the following bound for Grover search:

\begin{lemma}[Grover/BBHT search]
\label{le:grover}
Given  access to a predicate $f$ on domain $[N] = \{1. \ldots, N\}$, which is computable in constant time, there is a
bounded-error quantum procedure that can find an $x \in [N]$ such that $f(x)=1$ in time $\widetilde{O}(\sqrt{N/k})$, where $k$ is the number of $x \in [N]$ such that $f(x)=1$. If $f(x)=0$ for all $x \in [N]$, then the algorithm runs in time $\widetilde{O}(\sqrt{N})$ and returns $\perp$.
With a logarithmic number of repetitions, which is hidden in the $\widetilde{O}(*)$ notation, the probability of failure can be bounded by any inverse polynomial. 
\end{lemma}

An actionable outcome of  \textsc{GroverSearch}$(v)$ corresponds to an edge $(v,u)$ that results in a $\GE$ or a $\BE$.
Since a vertex can transition from \textsc{unlabeled} to \textsc{IN} to \out{}, this suggests that an edge incident to $v$ could transition from actionable, to non-actionable, and then back to actionable again.
This would be a problem for the analysis below since Grover takes less time to find an actionable edge incident to $v$ when there are many of them to select from.
Lemma~\ref{lem:actionable-monotonicity} (proven in
Section~\ref{sec:qpathrun}) establishes that between consecutive calls to
\textsc{GroverSearch}$(v)$, an edge incident to $v$ that is non-actionable cannot later become
actionable. Thus every edge responsible for a future successful call from
$v$ is already actionable at each earlier call from $v$.
This fact is used in the proof below.

\begin{lemma}
\label{le:qpath-runtime}
\textsc{QuantumPathSearch} runs in time
\(
\widetilde{O}(\sqrt{mn})
\).
\end{lemma}

\begin{proof}
Each time a vertex $v$ is at the top of the stack, the algorithm performs one
Grover-style search over its adjacency list of size $d_v$. Let $s_v$ denote the
number of successful calls from $v$, that is, calls that result in a $\GE$ or a $\BE$. The total number of searches that return a new unlabeled neighbor (resulting in a $\GE$) is at most $n$.
The total number of searches that result in a $\BE$ is also at most $n$.
Therefore, 
$\sum_v s_v \le 2n$.

An actionable outcome of  \textsc{GroverSearch}$(v)$ corresponds to an edge $(v,u)$ that results in a $\GE$ or a $\BE$.
Lemma~\ref{lem:actionable-monotonicity} establishes that between
consecutive calls to \textsc{GroverSearch}$(v)$, an edge incident to $v$  that is
non-actionable cannot become actionable.
Therefore, on the $j^{th}$  successful call from $v$, there are at least $s_v-j+1$ actionable remaining edges incident to $v$   among the $d_v$ candidates. 

By Lemma~\ref{le:grover}, the cost of the $j^{th}$   call from $v$ is $\widetilde{O}\!\left(\sqrt{d_v/(s_v-j+1)}\right)$.
Summing over the successful calls from $v$ gives
\[
\sum_{k=1}^{s_v} \widetilde{O}\!\left(\sqrt{\frac{d_v}{k}}\right)
=
\widetilde{O}\!\left(\sqrt{d_v} \sum_{k=1}^{s_v} \frac 1 {\sqrt{k}}\right)
=
\widetilde{O}\!\left(\sqrt{d_v\,s_v}\right).
\]
The last equation follows from the fact that $\sum_{j=1}^s j^{-1/2} = 2 \sqrt{s} +  O(1)$.
This identity can be verified by observing that the difference between the summation and the integral $\int^s_1 j^{-1/2} dj$ is $\pm O(1)$.

Using the Cauchy--Schwarz inequality,
\[
\sum_v \widetilde{O}\!\left(\sqrt{d_v\,s_v}\right)
\le
\widetilde{O}\!\left(\sqrt{\Bigl(\sum_v d_v\Bigr)\Bigl(\sum_v s_v\Bigr)}\right)
=
\widetilde{O}(\sqrt{mn}),
\]
since $\sum_v d_v=2m$ and $\sum_v s_v\le 2n$.

There is also at most one final unsuccessful search for each processed vertex,
whose total contribution is $\sum_v \widetilde{O}(\sqrt{d_v})
=
\widetilde{O}(\sqrt{mn})$.
Finally, the classical work spent on maintaining the temporary blossom structure
and on changing
\textsc{IN} vertices into \textsc{OUT} is $O(n)$ by
Lemma~\ref{le:cycle-cost}, and the union-by-size overhead is $O(n\log n)$ by
Lemma~\ref{le:union-by-size}. Hence the overall running time is
$\widetilde{O}(\sqrt{mn})$.
\end{proof}

\paragraph{Overview of task 3 of \textsc{QSearchOne}: Quantum Blossom Detection.}
After Task~2, Lemma~\ref{le:nomore} guarantees that the current eligible
contracted graph contains no augmenting path.  Task~3 therefore does not
need a second search algorithm.  Instead, we run the ordered search of
\textsc{QuantumPathSearch} from Task~1 again on the graph again with the updated
matching.  The search-local blossom family $\Omega'$ is initialized as a copy
of the current persistent family $\Omega$, and the search arrays
$UIO$,  $\mathrm{OutTime}$, and \(\mathrm{predEdge}[\cdot]\) are reinitialized exactly as in
Algorithm~\ref{alg:qpathsearch}.  The same calls to
\textsc{GroverSearch}, the same delayed OUT--OUT contraction rule, and the
same routine \textsc{ContractSearchBlossom} are used.

Since no eligible augmenting path exists when Task~3 begins, the
augmenting-path branch of Algorithm~\ref{alg:qpathsearch} is never taken.
Thus the search only performs grow steps, temporary blossom contractions,
and backtracking.  In contrast with Task~1, the blossoms created during this
second execution are not discarded.  If $\Omega_{\mathrm{old}}$ denotes the
persistent blossom family at the beginning of Task~3, then after the search
terminates we set
\[
    \Omega\gets\Omega'
\]
and initialize
\[
    z(B)\gets0
    \qquad
    \text{for every }B\in\Omega'\setminus\Omega_{\mathrm{old}}.
\]
The final $UIO$ labels and $\textsc{RootB}$ map are retained for the dual
adjustment in Task~4.

The ordered-DFS invariants (Lemmas \ref{lem:csp}, \ref{lem:blossomevent} and \ref{lem:qpath-dfs-invariants})  proved for Task~1 imply the termination
conditions required here.  Because no augmenting path is returned, every
OUT vertex is completely scanned when the search terminates, meaning that there has been a Grover Search from that vertex which returns $\perp$.  Consequently,
no eligible edge remains from an OUT vertex to an UNLABELED vertex in a
distinct root blossom, and every eligible OUT--OUT edge has both endpoints in
the same final root blossom.  The latter edges have therefore already been
absorbed by the search-blossom contractions.  These statements are proved
formally in Lemmas~\ref{le:q-forest} and~\ref{le:maxblossom} in
Section~\ref{sec:proofs}.  Since Task~3 executes the same ordered search as
Task~1, its running time is again $\widetilde{O}(\sqrt{mn})$ by
Lemma~\ref{le:qblossom-runtime}.

Finally the dual adjustment and blossom dissolution tasks (Tasks $4$ and $5$ in Figure \ref{fig:s1overview})
 are identical to the corresponding tasks of classical
\textsc{SearchOne}. The same linear-time bound therefore applies; see
Duan, Pettie, and Su~\cite{dps}.
A single transversal of all the trees in $\Omega$ takes $O(n)$ time and can be done
to update the entire $\textsc{RootB}$ array as the result of any dissolved blossoms.

Theorem \ref{th:qsearchone-correct} establishes 
\QSearchOne{} correctly performs all of the tasks as the classical \SearchOne{}
as outlined in Figure \ref{fig:s1overview}.
Theorem \ref{th:qsearchone-runtime} shows that 
the total running time of \QSearchOne{} is $\widetilde{O}(\sqrt{mn})$.

\section{Proofs and Pseudocode for \FVR{}}
\label{sec:proofs}
In this section, we provide the details, proofs, and pseudocode for \FVR{}.
Let $H = G_{\mathrm{elig}}/\Omega$
denote the eligible graph in which the persistent root blossoms of
$\Omega$ are contracted.  During
\textsc{QuantumPathSearch}, a search-local copy $\Omega'$ of $\Omega$
is maintained, and $\textsc{RootB}$ is the corresponding root-blossom
map.  Additional blossoms formed by the DFS are temporary blossoms in
$\Omega'$ and do not modify the persistent family $\Omega$.

The term ``blossom'' will refer to any blossom in the blossom structure $\Omega'$, which can include newly formed temporary blossoms or blossoms in $\Omega$.
We will use to \textbf{$H$-blossoms} to refer to any root blossom in $\Omega$. Thus, there is a one-to-one correspondence between $H$-blossoms and vertices in $H = G_{\mathrm{elig}}/\Omega$. Our algorithm which operates in the original graph $G$ maintains the invariant that all the vertices contained in the same $H$-blossom have the same label. Therefore, the label of an $H$-blossom is a well-defined notion. \textsc{QuantumPathSearch} returns a set of $H$-blossom augmenting paths. We prove in this section that those paths are disjoint and the set is maximal.

\begin{algorithm}[p]

\footnotesize
\caption{\textsc{QuantumPathSearch}$(G,M,y,z,\Omega)$}
\label{alg:qpathsearch}
\DontPrintSemicolon

\KwIn{$G,M,y,z$, persistent blossom family $\Omega$}
\KwOut{A maximal set $\Psi$ of pairwise blossom-disjoint eligible
augmenting paths in $H=G_{\mathrm{elig}}/\Omega$}

$\Omega' \gets$ a search-local copy of $\Omega$\;
$\textsc{RootB} \gets$ root-blossom map
induced by $\Omega'$\;

$\Psi\gets\emptyset$, $t\gets0$\;

\ForEach{$v\in V$}{
    $UIO[v] \gets \UNLABELED$\;
}

\ForEach{root blossom $B$ of $\Omega'$}{
    $\mathrm{predEdge}[B] \gets \perp$\;
}

\ForEach{free vertex $r$ with $UIO[r]=\UNLABELED$}{

    $S\gets[\,]$\;
    $B_r\gets\textsc{RootB}[r]$\;

    $t\gets t+1$\;
    $\mathrm{OutTime}[B_r]\gets t$\;

    \ForEach{$x\in V(B_r)$}{
        $UIO[x]\gets\OUTLAB$\;
        \textsc{Push}$(S,x)$\;
    }

    \While{$S\neq\emptyset$}{

        $v\gets\textsc{Top}(S)$\;
        $B_v\gets\textsc{RootB}[v]$\;

        $u\gets
        \textsc{GroverSearch}
        (v)$\;

        \eIf{$u=\perp$}{

            \textsc{Pop}$(S)$\;

        }{

            \eIf{$UIO[u]=\UNLABELED$}{
            \tcp{A $\GE$}                       

                $B_u\gets\textsc{RootB}[u]$\;

                $ \mathrm{predEdge}[B_u] \gets  (v,u)$\;

                \ForEach{$x\in V(B_u)$}{
                    $UIO[x]\gets\INLAB$\;
                }

                \ForEach{$x\in V(B_u)$}{

                    $q\gets M[x]$\;

                    \eIf{$q=\emptyset$}{
                        \tcp{Agmenting Path Found}
                        $\bar P\gets$ the contracted path obtained by
                        following
                        $\mathrm{predEdge}[\cdot]$
                        from $B_u$ until $\perp$ reached\;

                        $P\gets$ the path in
                        $H=G_{\mathrm{elig}}/\Omega$ obtained from
                        $\bar P$ by expanding every temporary blossom
                        of $\Omega'\setminus\Omega$ along its stored
                        defining cycle\;

                        $\Psi\gets\Psi\cup\{P\}$\;

                        \textbf{break} the while loop\;

                    }{

                        \If{
                            $\textsc{RootB}[q]
                            \neq
                            \textsc{RootB}[x]$
                            \textbf{and}
                            $y(x)+y(q)=w(x,q)$
                        }{

                        \tcp{New OUT Blossom Extends the Search}
                            $B_q\gets\textsc{RootB}[q]$\;
                            $\mathrm{predEdge}[B_q]\gets(x,q)$\;

                            $t\gets t+1$\;
                            $\mathrm{OutTime}[B_q]\gets t$\;

                            \ForEach{$p\in V(B_q)$}{
                                $UIO[p]\gets\OUTLAB$\;
                                \textsc{Push}$(S,p)$\;
                            }
                        }
                    }
                }

            }{

                \If{$UIO[u]=\OUTLAB$}{
                        \tcp{A $\BE$}

                    \textsc{ContractSearchBlossom}
                    $(v,u,S,UIO,\mathrm{predEdge},
                    \Omega',\textsc{RootB},
                    \mathrm{OutTime})$\;
                }
            }
        }
    }
}

\Return $\Psi$\;

\end{algorithm}

\begin{algorithm}[p]
\footnotesize
\caption{
\textsc{ContractSearchBlossom}
$(v,u,S,UIO,\mathrm{predEdge},
\Omega',\textsc{RootB},\mathrm{OutTime})$
}
\label{alg:contract-search-blossom}
\DontPrintSemicolon

\KwIn{
$v,u$ are \textsc{OUT} and
$\mathrm{OutTime}[\textsc{RootB}[u]] > \mathrm{OutTime}[\textsc{RootB}[v]]$.\;
~~~~~~~~~~~~The array $\mathrm{predEdge}$ encode the current
alternating search tree.
}

\KwOut{
Updated
$S,UIO,\mathrm{predEdge},
\Omega',\textsc{RootB},
\mathrm{Base}$, and $\mathrm{OutTime}$.
}

$B_v\gets\textsc{RootB}[v]$\;
$B_u\gets\textsc{RootB}[u]$\;

Following $\mathrm{predEdge}$ from $B_u$ toward the root, obtain the current
search-tree path
\[
B_v=C_0,C_1,\ldots,C_{2k}=B_u,
\]\;

\For{$i\gets1$ \KwTo $2k$}{
    $e_i\gets\mathrm{predEdge}[C_i]$\;
}

Create in $\Omega'$ a new temporary blossom $B$ whose ordered
cycle of children is \[
C_0,C_1,\ldots,C_{2k},
\]
and whose defining cycle edges are $e_1,e_2,\ldots,e_{2k},(u,v)$.
and whose base child is $C_0$\;
\tcp{The new  blossom inherits the search path of its  base.}
$\mathrm{predEdge}[B]\gets \mathrm{predEdge}[C_0]$\;
$\mathrm{OutTime}[B]\gets\mathrm{OutTime}[C_0]$\;
$C_{max} \gets
\displaystyle\arg\max_{C_i}
|V(C_i)|$
\tcp*{union-by-size representative}

Rename $C_{max}$ with a new label $C_{new}$\;
Rename $B$ with  label $C_{max}$\;
\;

\ForEach{$C_i\neq C_{new}$}{
    \ForEach{$a\in V(C_i)$}{
        $\textsc{RootB}[a]\gets C_{max}$\;
    }
}

\For{$j\gets k$ \KwTo $1$ \textbf{by} $-1$}{

    \tcp{$C_{2j-1}$ was IN and becomes exposed as OUT.}

    \ForEach{$x\in V(C_{2j-1})$}{
        $UIO[x]\gets\OUTLAB$\;
        \textsc{Push}$(S,x)$\;
    }
}

\end{algorithm}

\subsection{Proof of Correctness for Quantum Path Search}
\label{sec:qpathcorrect}

The DFS follows the ordered augmenting-path search of
Gabow and Tarjan~\cite{GabowTarjan91}.  In particular, an
OUT--OUT edge is used for a blossom step only when the target
blossom became OUT strictly after the source blossom, and the
blossoms that become OUT in a blossom step are subsequently explored
in the prescribed order.  The latter formulation and its DFS
invariants are also discussed by Gabow~\cite{Gabow2017}.
The proofs for correctness in this section closely follow the arguments given in
~\cite{GabowTarjan91}.

\begin{lemma}[Augmenting paths returned by
\textsc{QuantumPathSearch}]
\label{lem:qpath-valid}
Conditioned on the correctness of all calls to \textsc{GroverSearch},
Algorithm~\ref{alg:qpathsearch} constructs a set $\Psi$ of pairwise
$H$-blossom eligible augmenting paths.
The paths are disjoint in that no $H$-blossom is contained in more than one path.
Each path can be lifted to the original graph $G$ in time linear in its expanded length.
\end{lemma}

\begin{proof}
Consider a path $P$ added to $\Psi$.  The search starts at a free
$H$-blossom. Whenever the search moves from an OUT vertex to a previously
UNLABELED vertex, \textsc{GroverSearch} has returned an unmatched
eligible edge.  The newly reached $H$-blossom is therefore
labeled IN.  If it is matched and its external matched edge is
eligible, the search follows that unique matched edge and reaches a
new  $H$-blossom, $B$, which is then labeled OUT.  
Note that if $B$ was already reached previously in the search, the matched edge into $B$ would already have been included in the search structure from the opposite direction.
Consequently, the predecessor structure maintained
by the search alternates between unmatched and matched eligible
edges.

A temporary blossom contraction does not change this alternating
reachability.  The blossom is formed from an OUT--OUT bridge together
with the alternating search-tree path between its endpoints.  Hence
its defining cycle is odd and alternating, and its base is OUT.
The defining cycle is stored directly in $\Omega'$, and the new
temporary blossom inherits the predecessor
edge of its true base. Nested temporary blossoms can therefore be
expanded recursively from $\Omega'$.

If the search reaches a free $H$-blossom, following $\mathrm{predEdge}[\cdot]$ gives the current
contracted alternating path to the free root. Expanding the temporary
blossoms of $\Omega'\setminus\Omega$ along their stored defining
cycles returns an eligible augmenting path in $H$.

It remains to prove disjointness.  Once an $H$-blossom is
reached, it remains labeled for the remainder of
\textsc{QuantumPathSearch}.  A later grow step enters only an
UNLABELED blossom.  Moreover, an OUT blossom belonging to an
earlier free-root search has smaller $\mathrm{OutTime}$ than every
OUT blossom of the current search and therefore cannot satisfy the
OUT--OUT predicate of \textsc{GroverSearch}.  Thus a later search
cannot enter any blossom belonging to an augmenting path
already placed in $\Psi$.  Hence the paths of $\Psi$ are pairwise
$H$-blossom-disjoint.

Finally, since the paths in $\Psi$ are blossom-disjoint,  their total
expanded length is $O(n)$.  The stored predecessor and blossom-cycle
information therefore permits all path reconstructions in time
linear in the total output size.
\end{proof}

For every vertex $v$, we assign an $H$-blossom $b(v)$, which serves as a representative of the root blossom in $\Omega'$ to which $v$ belongs.
Thus if two vertices $v$ and $w$ have $b(v) = b(w)$, then they belong to the same root blossom in $\Omega'$.
We assume that the tree structure of $\Omega'$ always stores the base of each blossom as its leftmost child in the tree.
$b(v)$ can then be found by starting at blossom $\textsc{RootB}[v]$ in $\Omega'$
and following the pointers along the leftmost path in $\Omega'$ until an $H$-blossom is reached. Note that if $\textsc{RootB}[v] = B$, and $v$ is matched to a vertex outside $B$, then $b(v)$ is equal to the $H$-blossom to which $v$ belongs.

We say that a vertex $v$ in the graph has been \textbf{completely scanned} if there has been a Grover Search from $v$ that returns $\perp$.
If $v$ is the vertex at the top of the DFS stack then the \textbf{current search path} refers to the blossom path from $\textsc{RootB}[v]$ to the root of the current search tree.
If the temporary blossom structure is used to lift that blossom path to a path of $H$-blossoms (corresponding to a vertex path in $H = G/\Omega$), the resulting path is called the
\textbf{current $H$-blossom search path}. 

\begin{lemma}
\label{lem:csp}
    At any point in the execution of \textsc{QuantumPathSearch}, an OUT vertex in the current search tree that has not been completely scanned, belongs to a blossom in the current search path.
\end{lemma}

\begin{proof}
We maintain the invariant that if a blossom has any vertices on the stack then 
those vertices are all contiguous on the stack. Moreover, the order in which they appear
on the stack is in the same order as the current search path (from bottom to the root).
By induction on the number of steps.
When a search starts at a new free blossom,  all the vertices in that blossom are 
labeled OUT and placed on the stack, therefore becoming part of the current search path.
They are all from the same blossom and there is only one blossom in the current search path.

     A step resulting in a $\GE$, may create a new OUT blossom. If so all of the vertices from that blossom are placed on the stack and the new OUT blossom is the lowest blossom in the current search path. 

     Any vertices that transition from IN to OUT as a result of a $\BE$ will be in the blossom at the end of the current search path as a result of the blossom creation. 
     They will all be added  to the top of the stack.
    
    A blossom leaves the current search path only when all of the vertices in that blossom have been completely scanned and removed from the stack. 
    Therefore the invariant is maintained after each step. 
\end{proof}

\begin{lemma}
    \label{lem:blossomevent}
    If an edge $(v,w)$ is an eligible edge, where $v$ and $w$ are in the same search tree, both OUT, and
   in different root blossoms, then either
    \begin{enumerate}
        \item $\textsc{RootB}[v]$ is an ancestor of $\textsc{RootB}[w]$ in  the current DFS tree.
        \item $\textsc{RootB}[w]$ is an ancestor of $\textsc{RootB}[v]$ in  the current DFS tree.
    \end{enumerate}
    \end{lemma}

\begin{proof}
By induction on the number of steps.  
We need only consider actions that cause a vertex to become OUT.
Consider a $\GE$ resulting from a Grover Search at vertex $v$ that results in a new vertex $w$ becoming OUT. The vertex $w$ was unlabeled before the step and therefore cannot be
adjacent to a vertex that was completely scanned. Therefore if $w$ is adjacent to an OUT vertex in the same tree, by Lemma \ref{lem:csp}, that vertex must belong to a blossom in the current search path, meaning its blossom is an ancestor of $\textsc{RootB}[w]$ in the current search tree.
 A blossom step only contracts a
portion of the current DFS path and cannot create any cross edges between two blossoms where one is not a descendant of the other.
\end{proof}

\begin{corollary}
\label{cor:blossomevent}
    A Grover Search that returns an edge $(v,w)$ resulting in a $\BE$ must satisfy the condition that $\textsc{RootB}[v]$ is an ancestor of $\textsc{RootB}[w]$ in the DFS tree.
\end{corollary}

\begin{proof}
    The $\textsc{OutTime}[\cdot]$ value for the  blossoms in each search tree 
    increase downward in the tree.
    Therefore if a Grover Search at vertex $v$ returns an edge $(v,w)$, where 
    $\textsc{OutTime}[v] < \textsc{OutTime}[w]$, then $w$ cannot be in a previously searched tree and $\textsc{RootB}[w]$ cannot be an ancestor of $\textsc{RootB}[v]$. By Lemma \ref{lem:blossomevent}, the only remaining option is that $\textsc{RootB}[v]$ is an ancestor of $\textsc{RootB}[w]$ in the DFS tree.
\end{proof}

We are now ready to prove a stronger version of Lemma \ref{lem:csp}:

\begin{lemma}
\label{lem:chsp}
    At any point in the execution of \textsc{QuantumPathSearch}, an OUT vertex in the current search tree that has not been completely scanned, belongs to an $H$-blossom in the current $H$-blossom search path.
\end{lemma}

\begin{proof}
We maintain the invariant that if an $H$-blossom has any vertices on the stack then 
those vertices are all contiguous on the stack. Moreover, the order in which they appear
on the stack is in the same order as the current $H$-search path (from bottom to the root).

     A step resulting in a $\GE$, may create a new OUT blossom. If so all of the vertices from that blossom are placed on the stack and the new OUT blossom is the lowest blossom in the current search path. The blossom has not been a part of any merge operations, it is an $H$-blossom that is part of the current $H$-blossom search path. 

     Now consider a $\BE$, by Corollary \ref{cor:blossomevent}, $\textsc{RootB}[v]$ is an ancestor of $\textsc{RootB}[w]$ in the DFS tree. All of the vertices in the OUT blossoms in the path from $\textsc{RootB}[w]$  up to $\textsc{RootB}[v]$ have been completely scanned, since the algorithm has backtracked back to $\textsc{RootB}[v]$.
     The IN blossoms are all $H$-blossoms. Their vertices are added to the stack so that the highest $H$-blossom (in the search tree) is at the top of the stack. Note that any  alternating $H$-blossom path leading back to the root, will enter through one of these $H$-blossoms, head downward and then take the edge (w, v), back up to $\textsc{RootB}[v]$. Therefore, the stack order respects the order of the current $H$-search path. 
    
    An $H$-blossom leaves the current $H$-blossom search path only when all of the vertices in that $H$-blossom have been completely scanned and removed from the stack. 
    Therefore the invariant is maintained after each step. 
\end{proof}

\begin{corollary}
\label{cor:csp}
    When a free-root search terminates, every OUT vertex not belonging
    to an $H$-blossom in a path in $\Psi$ has been completely scanned.
\end{corollary}

\begin{proof}
    When a free-root search terminates, the path added to $\Psi$ is the current $H$-blossom search path. By Lemma \ref{lem:chsp}, an OUT vertex in the current search tree that has not been completely scanned, belongs to a blossom in the current $H$-blossom search path.
\end{proof}

\begin{lemma}
\label{lem:qpath-dfs-invariants}
    If an eligible edge $(v,w)$ joins two OUT vertices such that  $v$ and $w$ have both been completely scanned, then
    \[
        b(v)=b(w).
    \]
\end{lemma}

\begin{proof}
Consider the edge $(v,w)$ between two OUT vertices.
It must be the case that $v$ and $w$ belong to the same search tree. Otherwise, suppose that $v$ is in the first search tree that was created. Since $v$ was completely scanned, the edge
$(v,w)$ would have resulted in a $\GE$, which would have resulted in $v$ and $w$ being in the same search tree.

If $b(v) \neq b(w)$, then $v$ and $w$ are contained in different root blossoms,
and by Lemma \ref{lem:blossomevent} one must be a proper ancestor of the other in the search tree. Without loss of generality assume that $\textsc{RootB}[v]$
is a proper ancestor of $\textsc{RootB}[w]$.

Consider the possible outcomes of a Grover Search on vertex $v$.
If the edge $(v,w)$ is returned as a $\GE$, then $w$ is UNLABELED, and 
the scan performs the
corresponding grow step, at which point $w$ is labeled IN.  If $w$  later becomes OUT in a
blossom step, the edge $(v,w)$ lies on the
contracted alternating-tree path, so the two endpoints will then belong
to the same temporary blossom.
If the edge $(v,w)$ is returned as a $\BE$, then the blossoms containing $v$ and $w$ are contracted into a common temporary blossom.

The remaining possibility is that a Grover Search from $v$ never returns the edge $(v,w)$. If $v$ is completely scanned then, $w$  is either IN, or is
an OUT vertex for which the delayed test fails.  
Suppose $w$ is IN, and let $B$ be the lowest common ancestor of 
$\textsc{RootB}[v]$
and $\textsc{RootB}[w]$ in the search tree.
The search has backed up from every blossom on the search tree path from 
$\textsc{RootB}[w]$ back to $B$. The only way for $w$ to become OUT is for there
to be a blossom event that contracts the path from an ancestor of $B$ to a descendant of $\textsc{RootB}[w]$. This contradicts the assumption that
$\textsc{RootB}[v]$ is a proper ancestor of $\textsc{RootB}[w]$.
The last possibility is that $w$ is OUT. Since $\textsc{RootB}[v]$
is a proper ancestor of $\textsc{RootB}[w]$, this would result in a blossom event, placing $v$ and $w$ in the same blossom.
\end{proof}

\begin{lemma}[Maximality of \textsc{QuantumPathSearch}]
\label{lem:qpath-maximal}
Conditioned on the correctness of all calls to
\textsc{GroverSearch}, when Algorithm~\ref{alg:qpathsearch}
terminates, there is no eligible augmenting $H$-blossom path
that does not intersect with any path from $\Psi$.
\end{lemma}

\begin{proof}

Suppose for contradiction that, when the algorithm terminates, there
is an eligible augmenting path
\[
    Q=(B_0,B_1,\ldots,B_{2k+1}),
\]
where each $B_i$ in the path is an $H$ blossom that is not contained in
any path in $\Psi$.
The graph edge connecting $B_i$ to $B_{i+1}$ in the path will be denoted as $(v_i, u_{i+1})$, where $v_i \in B_i$ and $u_{i+1} \in B_{i+1}$.

We can extend the base function $b(\cdot)$, which is defined on vertices
to $H$-blossoms, since all the vertices in the same $H$-blossom have the same
value for the function $b(\cdot)$.
We by induction that prove for $j=0,\ldots,k$ that $B_{2j}$ is OUT and
\[
    b(B_{2j})=B_h
    \qquad\text{for some } h\le 2j.
\]
We will also prove that if the induction hypothesis holds for $j$, then
$B_{2j+1}$ cannot be a free $H$-blossom.

The base case $j=0$ holds because $B_0$ is free and does not belong
to $\Psi$, so the algorithm eventually starts a search from $B_0$.
The $H$-blossom will continue to be the base of its blossom, even after any contraction steps, so $b(B_0) = B_0$ will continue to hold throughout the search.

Suppose that $B_{2j}$ is OUT. Since $B_{2j}\notin V(\Psi)$,
Corollary \ref{cor:csp} implies that every vertex in $B_{2j}$ has been completely scanned,
including $v_{2j}$, the end point of the edge connecting $B_{2j}$ to $B_{2j+1}$.
Therefore $B_{2j+1}$ must be contained in the search structure.
If $B_{2j+1}$ is free and not in a path from $\Psi$, then it must be the
case that a search originated from $B_{2j+1}$. In this case $B_{2j+1}$ is
OUT. Moreover since a new search started from $B_{2j+1}$, it must be
the case that $b(B_{2j+1}) = B_{2j+1}$. However, we also have the
fact that $(v_{2j}, u_{2j+1})$ is an OUT-OUT edge with both endpoints completely scanned. Lemma \ref{lem:qpath-dfs-invariants} implies then that
$b(B_{2j+1}) = b(B_{2j}) = B_h$, for $h \le 2j$, a contradiction.
Therefore, $B_{2j+1}$ cannot be free and must be matched. If $j=k$,
this contradicts the fact that the terminal blossom $B_{2k+1}$ of $Q$
is free. It remains to establish the induction step when $j<k$.

If $B_{2j+1}$ is IN, then it lies outside every nontrivial
temporary search blossom, since every nontrivial temporary blossom
is OUT.  Hence the matched edge
$(v_{2j+1},u_{2j+2})$ crosses the boundary of the current temporary
blossom containing $B_{2j+2}$.  A matched blossom has at most one
matched edge leaving it, and that edge is incident to its base.
Therefore
\[
    b(u_{2j+2}) = b(B_{2j+2})=B_{2j+2}.
\]
Since $B_{2j+2}$ is OUT, the induction claim follows for $j+1$.

If $B_{2j+1}$ is OUT and $B_{2j+1}\notin V(\Psi)$, then
Corollary \ref{cor:csp} implies that all the vertices contained in $B_{2j+1}$ are completely scanned. Applying
Lemma \ref{lem:qpath-dfs-invariants} to the edge $(v_{2j},u_{2j+1})$ gives
\[
    b(B_{2j+1})=b(u_{2j+1})=b(v_{2j})=b(B_{2j})
\]
Therefore $b(B_{2j+1}) = B_h$ for some $h\le2j$. Since this common base is not $B_{2j+1}$, the
matched edge $(v_{2j+1},u_{2j+2})$ lies inside the same temporary
blossom. Consequently $B_{2j+2}$ is OUT and
\[
    b(B_{2j+2})=b(B_{2j+1})=B_h,
\]
so the induction claim again holds for $j+1$.

Therefore, by contradiction, there is no eligible augmenting $H$-blossom path
that does not intersect with any path from $\Psi$.
\end{proof}

\subsection{Proof for the Running Time of Quantum Path Search}
\label{sec:qpathrun}

The lemmas and proofs in this section are referenced
in the proof of Lemma \ref{le:qpath-runtime} which establishes that
the running time of 
\textsc{QuantumPathSearch} is
\(
\widetilde{O}(\sqrt{mn})
\).

\begin{lemma}[Monotonicity of actionable incidences]
\label{lem:actionable-monotonicity}
Let $\theta_1<\theta_2$ be two consecutive moments at which an
OUT vertex $v$ is scanned during one free-root search.  If an
incidence $(v,w)$ is not actionable from $v$ at time $\theta_1$,
then it does not become actionable from $v$ during
$(\theta_1,\theta_2)$.
\end{lemma}

\begin{proof}
The matching $M$ and the dual variables $y,z$ do not change during
Task~1.  Hence the Criterion~2 eligibility of an edge is fixed during
\textsc{QuantumPathSearch}.  Moreover, temporary contractions only
merge search blossoms.  It therefore suffices to consider changes in
the search labels and in the temporary-blossom structure.

Let $T$ be the value of the global counter $t$ when $v$ is scanned
at time $\theta_1$.  We first use the following stack invariant:
while $v$ remains on the stack, every OUT vertex $x$ above $v$ on
the stack either belongs to the same current temporary blossom as
$v$, or satisfies
\[
  \mathrm{OutTime}[\textsc{RootB}[x]]>T.
\]

This follows by induction on the search operations performed before
$v$ is scanned again.  A pop only removes a vertex above $v$.  A
grow step places on the stack a newly reached OUT vertex whose
$\mathrm{OutTime}$ is larger than $T$.  Finally, suppose a blossom
step is performed while $x$ is at the top of the stack.  By the
induction hypothesis, either $x$ is already in the temporary blossom
containing $v$, in which case every vertex newly absorbed by the
contraction is also in that blossom, or
\[
  \mathrm{OutTime}[\textsc{RootB}[x]]>T.
\]
In the latter case the new temporary blossom inherits the
$\mathrm{OutTime}$ of its base, namely that of
$\textsc{RootB}[x]$, so every newly activated OUT vertex placed on
the stack also belongs to a blossom whose $\mathrm{OutTime}$ is
larger than $T$.

Now consider a neighbor $w$ for which $(v,w)$ is not actionable at
$\theta_1$.  If $(v,w)$ is not eligible under Criterion~2, it remains
ineligible throughout Task~1.  If $v$ and $w$ are already in the
same temporary blossom, they remain so because temporary blossoms
are only contracted, never expanded.

It remains to consider a Criterion~2 eligible edge joining distinct
temporary blossoms.  Since $v$ is OUT, $w$ cannot be UNLABELED at
$\theta_1$, because then $(v,w)$ would be actionable.  Suppose first
that $w$ is IN.  If $w$ later becomes OUT before $\theta_2$, this
happens in a blossom step initiated from some OUT vertex $x$ above
$v$.  By the stack invariant, either $x$ lies in the temporary
blossom containing $v$, in which case the contraction also places
$w$ in that blossom and $(v,w)$ becomes internal, or
\[
  \mathrm{OutTime}[\textsc{RootB}[x]]>T.
\]
In the latter case every vertex newly exposed on the contracted
search-tree path was reached after time $T$ and therefore was
UNLABELED at $\theta_1$.  Thus this case cannot apply to a vertex
$w$ that was already IN at $\theta_1$.

Finally suppose that $w$ is OUT at $\theta_1$.  Since $(v,w)$ is
not actionable,
\[
  \mathrm{OutTime}[\textsc{RootB}[w]]
  \le
  \mathrm{OutTime}[\textsc{RootB}[v]].
\]
Before $v$ is scanned again, only vertices above $v$ are active.
The stack invariant implies that any contraction affecting such
vertices either merges them into the blossom containing $v$, making
$(v,w)$ internal, or has a base whose $\mathrm{OutTime}$ is larger
than $T$.  It cannot reverse the above ordering while leaving
$v$ and $w$ in distinct temporary blossoms.  Thus $(v,w)$ does not
become actionable before $\theta_2$.

\end{proof}

\begin{lemma}
\label{le:union-by-size}
The total cost of maintaining the root-blossom labels by union-by-size during
\textsc{ContractSearchBlossom} is $O(n\log n)$.
\end{lemma}

\begin{proof}
Whenever several current root-blossom classes are merged, we rename the labels
of the smaller classes to the largest class. Each time the label of a vertex is
renamed, the size of its class at least doubles. Since a class size never
exceeds $n$, each vertex is relabeled at most $\lfloor \log_2 n \rfloor$ times.
Hence the total number of relabelings over the whole blossom detection task is $O(n\log n)$.
\end{proof}

\begin{lemma}
\label{le:cycle-cost}
The total cost of maintaining the defining cycles and changing the
\textsc{IN} vertices into \textsc{OUT} during
\textsc{ContractSearchBlossom} is $O(n)$.
\end{lemma}

\begin{proof}

If a contraction uses a defining cycle containing $\ell$ current root blossoms, Algorithm~\ref{alg:contract-search-blossom} records only those $\ell$ current roots and replaces them by one new current root.  Thus the number of current root blossoms decreases by $\ell-1$.  Summed over all contractions, $\sum(\ell-1)=O(n)$; since there are also only $O(n)$ contractions, the total defining-cycle length is $O(n)$.

For the label changes, turning \textsc{IN} into \textsc{OUT} is monotone: once a
vertex becomes \textsc{OUT}, it remains \textsc{OUT} for the rest of
Algorithm~\ref{alg:qpathsearch}. Hence there are at most $O(n)$ such label changes.
Therefore the total classical work in these two tasks is $O(n)$.
\end{proof}

\subsection{Augmentation}

By Facts \ref{fact:aug1} and \ref{fact:aug2}, the augmenting paths in $G/\Omega$ found by \textsc{QuantumPathSearch} can be lifted to the original graph in time proportional to the sum of the lengths of the paths. 
The matching can then be updated by augmenting along those paths.

Note that Criterion 2 from \cite{dps} ensures that after augmentation, there are no eligible 
paths in the resulting path. This fact is proven in Lemma 2.3 in \cite{dps}.
We encapsulate that argument in the following lemma. 

\begin{lemma}
\label{le:nomore}
After the augmentation task, there is no eligible augmenting path in $G/\Omega$.
\end{lemma}

\begin{proof}Let $P$ be the set of maximal  blossom-disjoint 
paths found by Algorithm \ref{alg:qpathsearch}.
Suppose, for contradiction, that after augmenting along the paths in $P$, there is
an eligible augmenting path. 
Since the paths in $P$ are  maximal, the eligible augmenting path must share a blossom $B$ with some path in $P$.
Since under Criterion 2, edges inside and outside the matching have different properties, when
we augment along all the paths in $P$, the  edges contained in those paths all become ineligible. 
Therefore, the matched edge connecting $B$ to another root blossom is ineligible,
contradicting the fact that there is an eligible augmenting path containing blossom $B$.
\end{proof}

\subsection{Task 3: Quantum Blossom Detection}

Let $\Omega_{\mathrm{old}}$ be the persistent blossom family after Task~2.
By Lemma~\ref{le:nomore}, the eligible contracted graph
$G_{\mathrm{elig}}/\Omega_{\mathrm{old}}$ contains no augmenting path.
Task~3 reuses Algorithm~\ref{alg:qpathsearch} without modification: we start
a fresh execution on the updated matching, initialize a search-local family
$\Omega'$ as a copy of $\Omega_{\mathrm{old}}$, and reinitialize all search
labels and ordered-scan data exactly as in Task~1.

Any path that Algorithm~\ref{alg:qpathsearch} would add to $\Psi$ is, by
Lemma~\ref{lem:qpath-valid}, an eligible augmenting path in
$G_{\mathrm{elig}}/\Omega_{\mathrm{old}}$.  Lemma~\ref{le:nomore} therefore
implies that this branch is never reached during Task~3 and $\Psi$ remains
empty.  All operations that do occur---grow steps, matched steps,
OUT--OUT blossom contractions, activation of the former IN parts of a new
search blossom, and backtracking---are exactly the operations already defined
and analyzed for Task~1.

\begin{lemma}[Alternating forest in Task~3]
\label{le:q-forest}
During Task~3, the predecessor structure maintained by Algorithm \ref{alg:qpathsearch} is an alternating forest rooted at the free blossoms.  Moreover, every
eligible OUT--OUT edge whose endpoints lie in distinct current root blossoms
has both endpoints in the same alternating tree.
\end{lemma}

\begin{proof}
The first statement is the same alternating-tree invariant used in
Lemma~\ref{lem:qpath-valid}: grow edges from OUT to newly reached blossoms are
eligible unmatched edges, and every matched IN blossom is followed through
its unique eligible external matched edge to an OUT blossom.  Temporary
blossom contractions preserve this alternating reachability.

For the second statement, suppose that an eligible OUT--OUT edge $(u,v)$
joined two distinct alternating trees.  The alternating tree path from the
free root of the first tree to $u$, followed by $(u,v)$, followed by the
reverse alternating tree path from $v$ to the free root of the second tree,
would form an eligible augmenting path after expanding any temporary search
blossoms along their stored defining cycles.  This would give an eligible
augmenting path in $G_{\mathrm{elig}}/\Omega_{\mathrm{old}}$, contradicting
Lemma~\ref{le:nomore}.  Hence the two endpoints must lie in the same
alternating tree.
\end{proof}

\begin{lemma}[Exhaustion of Task~3]
\label{le:maxblossom}
When Task~3 terminates:
\begin{enumerate}
    \item every OUT vertex is completely scanned;
    \item there is no eligible edge from an OUT vertex to an UNLABELED vertex
    in a distinct current root blossom; and
    \item there is no eligible OUT--OUT edge whose endpoints lie in distinct
    current root blossoms.
\end{enumerate}
\end{lemma}

\begin{proof}
As observed above, $\Psi=\emptyset$ throughout Task~3.  Corrolary \ref{cor:csp} therefore implies that every OUT vertex
is completely scanned when its free-root search terminates.  This proves
Item~1.

For Item~2, suppose that an eligible edge $(v,u)$ remained at termination
with $v$ OUT, $u$ UNLABELED, and the endpoints in distinct current root
blossoms.  Since $v$ is completely scanned, its final call to
\textsc{GroverSearch} returned $\perp$.  But if $(v,u)$ were still an
eligible edge to an UNLABELED vertex in a distinct root blossom, then
$(v,u)$ would be an actionable grow event and hence a marked item for
that call, a contradiction.

For Item~3, let $(v,u)$ be an eligible OUT--OUT edge at termination.  By
Item~1, both endpoints are completely scanned, so
Lemma~\ref{lem:qpath-dfs-invariants} gives
\[
    b(v)=b(u),
\]
where $b(x)$ is the base of the current search blossom containing $x$.
Therefore $u$ and $v$ lie in the same final root blossom of $\Omega'$.  Hence
no eligible OUT--OUT edge remains between distinct current root blossoms.
\end{proof}

At the end of the search, the search-local blossoms are promoted to
persistent blossoms.  Namely, we set
\[
    \Omega\gets\Omega',
\]
retain the final $\textsc{RootB}$ and $UIO$ arrays, and set
\[
    z(B)\gets0
    \qquad
    \text{for every }B\in\Omega'\setminus\Omega_{\mathrm{old}}.
\]
Each promoted blossom is a valid odd alternating blossom because it was
created by \textsc{ContractSearchBlossom}, whose defining cycle and true base
are exactly the temporary-blossom structure already used in the correctness
proof of Task~1.  Thus the promoted family remains laminar and has the same
blossom representation required by the subsequent dual-adjustment and
dissolution tasks.

\begin{lemma}
\label{le:qblossom-runtime}
Task~3 runs in time
\[
    \widetilde{O}(\sqrt{mn}).
\]
\end{lemma}

\begin{proof}
Task~3 executes the same ordered search as
Algorithm~\ref{alg:qpathsearch}; the augmenting-path output branch is simply
never reached.  Its search and blossom-contraction costs are therefore
bounded by Lemma~\ref{le:qpath-runtime}.  Promoting the newly created search
blossoms and initializing their $z$-values costs $O(n)$ additional classical
work.  Hence the total running time is
$\widetilde{O}(\sqrt{mn})$.
\end{proof}

\subsection{Dual adjustment and blossom dissolution}

The final two tasks of \QSearchOne{} are identical to the
corresponding tasks of classical \textsc{SearchOne}. Let $V_{\textsc{OUT}}$ and
$V_{\textsc{IN}}$ denote the sets of vertices reachable from the free vertices by an even or
odd length alternating path in the contracted graph, respectively, and let $\rho_\Omega$ denote the set
of root blossoms. The routine performs one dual adjustment:
\begin{align*}
y(v) &\gets y(v)-1 && \forall v\in V_{\textsc{OUT}},\\
y(v) &\gets y(v)+1 && \forall v\in V_{\textsc{IN}},\\
z(B) &\gets z(B)+2 && \forall B\in \rho_\Omega \text{ with } B\subseteq V_{\textsc{OUT}},\\
z(B) &\gets z(B)-2 && \forall B\in \rho_\Omega \text{ with } B\subseteq V_{\textsc{IN}}.
\end{align*}
After this step, every root blossom $B$ with $z(B)=0$ is dissolved.

\begin{lemma}
\label{le:qsearchone-adjust}
The dual-adjustment and blossom-dissolution of
\QSearchOne{} take $O(n)$ time.
\end{lemma}

\begin{proof}
These tasks are identical to the corresponding tasks of classical
\textsc{SearchOne}. The same linear-time bound therefore applies; see
Duan, Pettie, and Su~\cite{dps}.
A single transversal of all the trees in $\Omega$ takes linear time and can be done
to update the entire $\textsc{RootB}$ array as the result of any dissolved blossoms.
\end{proof}

\subsection{Correctness and running time}

\begin{theorem}
\label{th:qsearchone-correct}

\QSearchOne{} correctly performs all of the tasks as the classical \SearchOne{}
as outlined in Figure \ref{fig:s1overview}.
\end{theorem}

\begin{proof}
By Lemmas~\ref{lem:qpath-valid} and ~\ref{lem:qpath-maximal}, Task~1 finds a maximal set of eligible vertex-disjoint
augmenting paths in the contracted graph.
After lifting the paths to the original graph and augmenting along the  set of paths, Lemma~\ref{le:nomore} shows that no eligible
augmenting path remains.

 By Lemmas~\ref{le:maxblossom}, after augmenation, every discovered eligible
\textsc{OUT}--\textsc{OUT} edge lies inside a single
alternating tree, so it defines a new eligible odd cycle rather than a new
augmenting path. 

Thus Task~3 has the same kind of blossom-family update required by the blossom-shrinking
task of classical \textsc{SearchOne}. Thus, the blossom structure $\Omega$ is maintained in the same 
way as the classical \textsc{SearchOne}. The $\textsc{RootB}$ array, reflecting the root blossom for each vertex is
updated along with each new blossom created.

Finally, the dual-adjustment and zero-$z$ dissolution tasks are identical to
the classical ones. The $\textsc{RootB}$ array introduced for \QSearchOne{} can be updated in $O(n)$ time.
Therefore the input-output behavior is the same as for
\textsc{SearchOne}.
\end{proof}

\begin{theorem}
\label{th:qsearchone-runtime}
The total running time of \QSearchOne{} is
\[
\widetilde{O}(\sqrt{mn}).
\]
\end{theorem}

\begin{proof}
By Lemma~\ref{le:qpath-runtime}, the path-search task costs
$\widetilde{O}(\sqrt{mn})$. By Lemma~\ref{le:qblossom-runtime}, the
blossom-detection also costs $\widetilde{O}(\sqrt{mn})$. The
dual-adjustment and blossom-dissolution tasks take $O(n)$ time by
Lemma~\ref{le:qsearchone-adjust}. Therefore
\[
\widetilde{O}(\sqrt{mn})+\widetilde{O}(\sqrt{mn})+O(n)
=
\widetilde{O}(\sqrt{mn}),
\]
as claimed.
\end{proof}

\section{Our Full Quantum-Classical MWPM Algorithm}
\label{sec:full}

Section~\ref{sec:proofs} established that \QSearchOne{} correctly performs
the tasks of \SearchOne{} shown in Figure~\ref{fig:s1overview}, with running
time \(\widetilde{O}(\sqrt{mn})\). We now show how this routine is incorporated
into the \Liq{} framework and derive the correctness and running time of our
quantum-classical algorithm.

\subsection{The \textsc{QLiquidationist} Algorithm Overview}
\label{sec:liquidationist-overview}

The previous sections defined the blossom structure, the dual variables \(y\)
and \(z\), the quantity \(yz(e)\), and the eligibility condition used during
\FVR{}. We now describe the remaining parts of the classical \Liq{} framework  of Duan, Pettie, and Su~\cite{dps} 
that are needed to place \QSearchOne{} inside the full scaling algorithm.

The \Liq{} algorithm processes the integer
edge weights over \(O(\log (nW))\) scales. At the beginning of each scale we perform
the Initialization and Scaling steps of Duan--Pettie--Su unchanged.

Let $\widehat w$ denote the original integer edge-weight function and define
the extended weights
\[
    \bar w(e)=\left(\frac n2+1\right)\widehat w(e).
\]
The \Liq{} algorithm processes the bits of $\bar w$ over
\[
    L=\left\lceil
        \log\!\left(\left(\frac n2+1\right)W\right)
      \right\rceil
    =O(\log(nW))
\]
scales.  Initially,
\[
    y\gets0,\qquad z\gets0,\qquad w\gets0,\qquad
    \Omega\gets\emptyset .
\]
At the beginning of every scale, the matching is reinitialized and the
blossoms from the preceding scale are treated as old blossoms; the current
weights and dual variables are then scaled as described below.

The search used in the other procedures of the algorithm employs two additional
notions of eligibility.

\vspace{.1in}
\noindent{\bf Criterion 1:} An edge \(e\) is \emph{eligible under Criterion~1} if it is tight:
\[
    yz(e)=w(e).
\]

\vspace{.1in}
\noindent{\bf Criterion 3:} An edge \(e\) is \emph{eligible under Criterion~3} if
\[
    yz(e)=w(e)
    \qquad\text{or}\qquad
    yz(e)=w(e)-2.
\]

Criterion~1 is used during \SBL{}. Criterion~2, defined in Section~2, is used
during \FVR{}. Criterion~3 is used during the final \Finalization{} procedure.

A blossom inherited from the preceding scale is called an \emph{old blossom}.
Given a threshold parameter \(\tau\), an old blossom is \emph{large} if it
contains at least \(\tau\) vertices and \emph{small} otherwise. To
\emph{liquidate} an old blossom \(B\), the algorithm performs
\[
    y(u)\gets y(u)+\frac{z(B)}{2}
    \qquad\text{for every }u\in B,
\]
sets \(z(B)\gets0\), and dissolves \(B\) from the blossom family. Liquidation
therefore transfers the dual value of the inherited blossom to its constituent
vertices.

Each scale of the \Liq{} algorithm consists of the following procedures.

\begin{enumerate}
    \item \textbf{Initialization and scaling.}
    Let $w',y',z'$ and $\Omega'$ denote the weight function, dual variables,
    and blossom family inherited from the preceding scale.  Initialize
    \[
        M\gets\emptyset,\qquad
        \Omega'\gets\Omega,\qquad
        \Omega\gets\emptyset,
    \]
    and retain copies
    \[
        w'\gets w,\qquad y'\gets y,\qquad z'\gets z.
    \]
    If $b_i(e)\in\{0,1\}$ is the $i$th bit of $\bar w(e)$, perform the
    scaling updates
    \[
        w(e)\gets2\bigl(w'(e)+b_i(e)\bigr)
            \qquad\text{for every }e\in E,
    \]
    \[
        y(u)\gets2y'(u)+3
            \qquad\text{for every }u\in V,
    \]
    and
    \[
        z(B')\gets2z'(B')
            \qquad\text{for every }B'\in\Omega'.
    \]
    Thus the blossoms inherited from the preceding scale are precisely the
    old blossoms processed in the subsequent liquidation steps.

    \item \textbf{Large-blossom liquidation.}
    Every large old blossom is liquidated.

    \item \textbf{Reweighting.}
    Each edge is reweighted according to
    \[
        w(u,v)\gets w(u,v)-y(u)-y(v),
    \]
    after which all vertex dual variables are reset to zero.

    \item \textbf{\SBL{}.}
    Every remaining old small blossom is liquidated. For each maximal old
    small blossom \(B\), the algorithm then repeatedly runs the classical
    Edmonds-search routine \textsc{PQSearch} from selected free vertices of
    \(B\). These searches use Criterion~1.

    \item \textbf{\FVR{}.}
    Once all inherited blossoms have been removed, the algorithm switches to
    Property~2. The \FVR{} procedure then performs exactly \(\tau\) successive
    calls to \SearchOne{} under Criterion~2. In our algorithm, these calls are
    implemented by the \QSearchOne{} routine developed in
    Sections~\ref{sec:result} and~\ref{sec:proofs}.

    \item \textbf{Perfection.}
    The free dummy vertices inherited from the preceding scale are deleted.
    For every remaining free vertex \(u\), the algorithm creates a new dummy
    vertex \(\widehat{u}\), sets
    \[
        y(\widehat{u})\gets\tau,
        \qquad
        w(u,\widehat{u})\gets0,
    \]
    and inserts \((u,\widehat{u})\) into the matching. Consequently, the
    augmented graph again has a perfect matching at the end of the scale.
\end{enumerate}

After the final scale, all dummy vertices are deleted. The algorithm then
repeatedly performs classical Edmonds searches under Criterion~3 until no free
vertices remain. This last stage is the \Finalization{} procedure.

Our algorithm retains this outer \Liq{} structure. The inherited-blossom
processing, \SBL{}, Perfection, and \Finalization{} remain classical and are 
performed exactly as described by  Duan, Pettie, and Su in~\cite{dps}, while the
complete \FVR{} procedure is realized through our quantum-compatible
\QSearchOne{}.

\begin{theorem}
\label{th:qliquidationist-correct}
Conditioned on every quantum search call returning a correct result,
\textsc{QLiquidationist} returns a maximum-weight perfect matching.
\end{theorem}

\begin{proof}
    
Figure~\ref{fig:s1overview} lists the five tasks of
\QSearchOne{}/\SearchOne{}: finding augmenting paths, augmenting along them,
detecting and contracting new blossoms, adjusting the dual variables, and
dissolving zero-valued blossoms. By Theorem~\ref{th:qsearchone-correct},
\QSearchOne{} performs the same tasks as classical \SearchOne{} and produces
the same type of output required by the \Liq{} algorithm: the matching,
blossom family, dual variables, and dissolved blossoms are updated exactly as
required by the corresponding \SearchOne{} call. 

The classical \Liq{} algorithm returns a maximum-weight perfect matching.
More precisely, every scale preserves the required scaled complementary-slackness
invariants, and the final clean-up phase returns an optimal perfect matching. Therefore, by the same analysis as in Duan, Pettie, and Su~\cite{dps}, \textsc{QLiquidationist} returns a maximum-weight perfect matching.
\end{proof}

\subsection{Running-time contributions}
\label{sec:full-runtime}

We next bound the work outside a single call to \QSearchOne{}.
The analyses of \SBL{} and \Finalization{} use the following running-time
bound for the classical Edmonds-search routine.

\begin{lemma}[\textsc{PQSearch} running time, imported from~\cite{dps}]
\label{le:dps-pq}
Let \(H\) be the graph on which one invocation of \textsc{PQSearch} is
performed, and let \(n(H)\) and \(m(H)\) denote its numbers of vertices and
edges. Then the invocation takes
\[
    \widetilde{O}\bigl(m(H)+n(H)\bigr)
\]
time, including the priority-queue operations, blossom contractions, and
maintenance of the classical search data structures.
\end{lemma}

For a \textsc{PQSearch} call restricted to a blossom \(B\), we use this bound
with \(H=G[B]\). For \Finalization{}, we use it with \(H=G\). In the dense
regime of the main theorem, \(m\ge n^{3/2}\), so
\(\widetilde{O}(m+n)=\widetilde{O}(m)\) for \Finalization{}.

\paragraph{Cost of \SBL{}.}

The alternative analysis of \SBL{}, which we present here, enables us to
optimize \(\tau\) and achieve a quantum advantage. Under the original analysis,
the cost is \(\widetilde{O}(\tau m\log W)\).

\begin{lemma}
\label{le:detail}
Across all scales, the total time spent in \SBL{} is
\[
    \widetilde{O}(n\tau^2\log W).
\]
\end{lemma}

\begin{proof}
Fix one scale and let \(B\) be a maximal old small blossom. Let \(n(B)\) and
\(m(B)\) denote the number of vertices and edges, respectively, in the
subgraph induced by \(B\). Since \(B\) is small,
\[
    n(B)<\tau.
\]
Because the input graph is simple,
\[
    m(B)
    \le
    \binom{n(B)}{2}
    =
    O(\tau^2).
\]

By Lemma~\ref{le:dps-detail}, \SBL{} makes \(O(n(B))\) calls to
\textsc{PQSearch} restricted to \(B\). Each call costs
\[
    \widetilde{O}\bigl(m(B)+n(B)\bigr)
\]
by Lemma~\ref{le:dps-pq}. Since \(m(B)=O(\tau^2)\) and \(n(B)<\tau\), this is
\(\widetilde{O}(\tau^2)\). Hence the total cost associated with \(B\) is
\[
    \widetilde{O}\bigl(n(B)\tau^2\bigr).
\]

The maximal old small blossoms considered during one scale are pairwise
vertex-disjoint. Therefore,
\[
    \sum_B n(B)\le n,
\]
and the total \SBL{} cost in one scale is
\[
    \widetilde{O}(n\tau^2).
\]
There are \(\widetilde{O}(\log W)\) scales, so the total cost over the full algorithm is
\[
    \widetilde{O}(n\tau^2\log W).
\]
\end{proof}

\paragraph{Cost of \FVR{}.}

\begin{lemma}
\label{le:fvr-cost}
Across all scales, the total time spent in \FVR{} is
\[
    \widetilde{O}(\tau\sqrt{mn}\log W).
\]
\end{lemma}

\begin{proof}
Each scale contains exactly \(\tau\) calls to \QSearchOne{}. By
Theorem~\ref{th:qsearchone-runtime}, each call takes
\(\widetilde{O}(\sqrt{mn})\) time. Thus the cost in one scale is
\[
    \widetilde{O}(\tau\sqrt{mn}),
\]
and multiplying by the \(\widetilde{O}(\log W)\) scales gives the claimed bound.
\end{proof}

\paragraph{Cost of \Finalization{}.}

\begin{lemma}
\label{le:clean}
The total time spent in \Finalization{} is
\[
    \widetilde{O}\!\left(\frac{nm}{\tau}\log W\right).
\]
\end{lemma}

\begin{proof}
By Lemma~\ref{le:dps-free}, after the \FVR{} procedure of one scale, at most
\(O(n/\tau)\) free vertices remain. The proof of this classical bound depends
only on the maintained properties of the \Liq{} framework and on the
postcondition of \SearchOne{}. Since Theorem~\ref{th:qsearchone-correct} shows
that \QSearchOne{} performs the same tasks as \SearchOne{}, the bound
continues to hold for \textsc{QLiquidationist}.

After all scales have been completed and the dummy vertices have been deleted,
the total number of free vertices that may need to be processed by
\Finalization{} is therefore
\[
    \widetilde{O}\!\left(\frac{n}{\tau}\log W\right).
\]
Each successful \textsc{PQSearch} call decreases the number of free vertices.
By Lemma~\ref{le:dps-pq}, a \textsc{PQSearch} call during \Finalization{} costs
\[
    \widetilde{O}(m+n).
\]
In the dense regime considered in this paper, \(m\ge n^{3/2}\), so this is
\(\widetilde{O}(m)\). There may also be one final unsuccessful call certifying
that no further augmentation is possible; its \(\widetilde{O}(m)\) cost is
absorbed by the displayed bound. It follows that the total cost of
\Finalization{} is
\[
    \widetilde{O}\!\left(\frac{nm}{\tau}\log W\right).
\]
\end{proof}

\paragraph{Remaining classical and memory-update costs.}

At the beginning of each scale, the scaled edge weights and the corresponding
QRAM-readable adjacency data are updated. The large-blossom liquidation,
reweighting, Perfection, and other classical bookkeeping steps require at most
\(\widetilde{O}(m+n)\) work per scale. Consequently, their total contribution
is
\[
    \widetilde{O}\bigl((m+n)\log W\bigr).
\]
In the dense regime considered in this paper, \(m\ge n^{3/2}\), and this
simplifies to
\[
    \widetilde{O}(m\log W).
\]
This term includes the initial storage of the graph, all per-scale weight
updates, and the sequential updates to arrays that are later read coherently.

Combining the preceding bounds gives the following running time for an
arbitrary choice of \(\tau\).

\begin{lemma}
\label{le:total-cost}
The total running time of \textsc{QLiquidationist} is
\[
    \widetilde{O}\!\left(
        \left(
            n\tau^2
            +
            \tau\sqrt{mn}
            +
            \frac{nm}{\tau}
            +
            m
        \right)\log W
    \right).
\]
\end{lemma}

\begin{proof}
The four terms are, respectively, the total costs of \SBL{}, \FVR{},
\Finalization{}, and the remaining classical and memory-update operations.
They are bounded by Lemmas~\ref{le:detail}, \ref{le:fvr-cost}, and
\ref{le:clean}, together with the preceding bookkeeping analysis.
\end{proof}

\subsection{Proof of the main theorem}
\label{sec:main-theorem-proof}

\begin{theorem}
\label{th:main}
Let \(G=(V,E)\) be a simple graph that admits a perfect matching, where
\[
    n=|V|,
    \qquad
    m=|E|,
\]
and the integer edge weight magnitudes are at most \(W\). For every constant \(c>0\),
\textsc{QLiquidationist} can be implemented so that, with probability at least
\(1-n^{-c}\), it returns a maximum-weight perfect matching in time
\[
    \widetilde{O}(nm^{2/3}\log W)
\]
in the dense regime
\[
    n^{3/2}\le m\le n^2.
\]
\end{theorem}

\begin{proof}
Each Grover/BBHT invocation is amplified so that its failure probability is a
sufficiently small inverse polynomial in the input parameters. The number of
quantum search calls made over all scales is polynomial in the input size.
Therefore, by a union bound, the probability that any quantum search returns
an incorrect result is at most \(n^{-c}\), for any fixed constant \(c>0\),
after suitable amplification.

Conditioned on the complementary event, every quantum search returns a valid
marked item when one exists and returns \(\perp\) only when no marked item
exists. Theorem~\ref{th:qliquidationist-correct} then implies that the
algorithm returns a maximum-weight perfect matching.

It remains to optimize the running time in Lemma~\ref{le:total-cost}. Set
\[
    \tau=\left\lceil m^{1/3}\right\rceil.
\]
Ignoring constant factors caused by rounding, the first and third terms satisfy
\[
    n\tau^2
    =
    O(nm^{2/3})
\]
and
\[
    \frac{nm}{\tau}
    =
    O(nm^{2/3}).
\]

For the \FVR{} term,
\[
    \tau\sqrt{mn}
    =
    O\!\left(m^{1/3}\sqrt{mn}\right)
    =
    O\!\left(n^{1/2}m^{5/6}\right).
\]
Since \(G\) is simple, \(m\le n^2\), and hence
\[
    m^{1/6}\le n^{1/3}\le n^{1/2}.
\]
It follows that
\[
    n^{1/2}m^{5/6}
    =
    nm^{2/3}\cdot\frac{m^{1/6}}{n^{1/2}}
    \le
    nm^{2/3}.
\]
Similarly,
\[
    m\le nm^{2/3},
\]
so the classical and memory-update term is also absorbed. Therefore,
\[
    n\tau^2
    +
    \tau\sqrt{mn}
    +
    \frac{nm}{\tau}
    +
    m
    =
    O(nm^{2/3}),
\]
and Lemma~\ref{le:total-cost} gives
\[
    T_{\textsc{QLiquidationist}}
    =
    \widetilde{O}(nm^{2/3}\log W).
\]

The best known classical combinatorial bound is
\[
    \widetilde{O}(m\sqrt{n}\log W).
\]
The ratio of the classical bound to our bound is
\[
    \frac{m\sqrt{n}}{nm^{2/3}}
    =
    \left(\frac{m}{n^{3/2}}\right)^{1/3}.
\]
Thus the two bounds have the same asymptotic order at
\(m=\Theta(n^{3/2})\), while our bound is asymptotically smaller when
\[
    m=\omega(n^{3/2}).
\]
\end{proof}

\section{Discussion and Open Problems}

\label{sec:discussion}

In \FVR{}, the dominant work comes from repeated local scans over adjacency lists subject
to eligibility and label constraints. These scans fit naturally into the Grover-search
paradigm: the search space is explicit, the predicate is local and reversible, and the
surrounding routine can be organized so that the classical invariants are preserved.

By contrast, \SBL{} and \Finalization{} are implemented through a priority-queue-driven
Edmonds search, called \textsc{PQSearch} in the \Liq{} algorithm. A call to
\textsc{PQSearch} does not merely ask for some eligible edge or some useful event.
Rather, it maintains a global set of candidate events with priorities, repeatedly extracts
the minimum-key event, discards stale events, and updates the search structure. Thus the
bottleneck is not just local edge discovery; it is the dynamic scheduling of many
competing events. A direct Grover search over adjacency lists does not by itself replace
this priority-queue mechanism. This is why the present paper accelerates \FVR{} while
leaving \SBL{} and \Finalization{} classical.

The most immediate open problem is to design a quantum analogue of the event-management tool in \textsc{PQSearch}. Any improvement over the current $\widetilde{O}(m)$ classical
cost of \textsc{PQSearch} would improve the cost of \SBL{} and \Finalization{}, and could
lead to quantum speedups for MWPM beyond the dense regime considered in this paper.

A second possible direction is approximate weighted matching. There are scaling-based
approximation algorithms for maximum-weight matching in general graphs, including the
linear-time $(1-\epsilon)$-approximation framework in \(O(m \epsilon^{-1} \log \epsilon^{-1})\) time~\cite{DP14}. These algorithms
have weaker output requirements than exact MWPM, and therefore may avoid some of the
event-management bottlenecks that arise in the exact \Liq{} framework. It would
be interesting to understand whether Grover-style local search can improve the dependence
on $m$ or $\epsilon$ in such approximate algorithms while preserving their primal-dual
structure.

\section*{Acknowledgements} The first and second authors were supported by the National Science Foundation under Grant No.~NSF-2533522. The third author was supported by the National Science Foundation under Grant No.~NSF-2212129.

\bibliographystyle{plain}
\bibliography{mwpm_citations}

\clearpage
\appendix
\section{Previous Work}
\subsection{Classical algorithms on exact MWPM}
We briefly review the main lines of work on exact maximum-weight perfect matching in
general graphs. Table~\ref{tab:mwpm-classical} summarizes representative classical bounds.

\begin{table}[h]
    \centering
    \small
    \setlength{\tabcolsep}{5pt}
    \renewcommand{\arraystretch}{1.15}
    \begin{tabularx}{\linewidth}{|r|l|>{\raggedright\arraybackslash}X|}
        \hline
        \textbf{Year} & \textbf{Reference} & \textbf{Time bound (MWPM) / assumptions} \\
        \hline
        1965 & Edmonds~\cite{Edmonds65poly} & $O(mn^2)$ \\
        \hline
        1974 & Gabow~\cite{Gabow74thesis} & $O(n^3)$ \\
        \hline
        1976 & Lawler~\cite{Lawler76} & $O(n^3)$ \\
        \hline
        1976 & Karzanov~\cite{Karzanov76} & $O(n^3 + mn\log n)$ \\
        \hline
        1978 & Cunningham--Marsh~\cite{CunninghamMarsh78} & $\mathrm{poly}(n)$ \\
        \hline
        1982 & Galil--Micali--Gabow~\cite{GalilMicaliGabow82} & $O(mn\log n)$ \\
        \hline
        1985 & Gabow~\cite{Gabow85scaling} & $O(mn^{3/4}\log W)$ (integer weights) \\
        \hline
        1989 & Gabow--Galil--Spencer~\cite{GabowGalilSpencer89} &
        $O\!\left(mn\log\log\log_d n + n^2\log n\right)$, where $d = 2 + m/n$ \\
        \hline
        1990 & Gabow~\cite{Gabow90soda} & $O(mn + n^2\log n)$ \\
        \hline
        1991 & Gabow--Tarjan~\cite{GabowTarjan91} &
        $O\!\left(m\sqrt{n\,\alpha(n,m) \log n}\,\log(nW)\right)$ (integer weights) \\
        \hline
        2012 & Cygan--Gabow--Sankowski~\cite{CyganGabowSankowski12} &
        $\widetilde{O}(Wn^\omega)$ (randomized; integer weights) \\
        \hline
        2017 & Duan--Pettie--Su~\cite{dps} &
        $O(m\sqrt{n}\,\log(nW))$ (integer weights) \\
        \hline
    \end{tabularx}
    \caption{Selected exact algorithms for maximum-weight perfect matching in general graphs.
    Here $n=|V|$, $m=|E|$, $W$ upper-bounds the magnitude of integer edge weights,
    $\alpha(\cdot,\cdot)$ is the inverse Ackermann function, and $\omega$ is the matrix
    multiplication exponent.}
    \label{tab:mwpm-classical}
\end{table}

\paragraph{Edmonds' primal--dual framework.}
The starting point is Edmonds' blossom algorithm, which can be interpreted as a
primal--dual method for the matching LP and its dual~\cite{Edmonds65poly,Edmonds65ptf}. The algorithm maintains a feasible
matching together with dual variables and grows alternating forests along tight edges.
When the search encounters an odd cycle obstructing augmentation, it contracts that cycle
into a blossom. Dual adjustments preserve feasibility while creating new tight edges, and
the process continues until an augmenting path is found or optimality is certified.

\paragraph{Scaling algorithms for integer weights.}
For integer weights, the standard acceleration is weight scaling. One reveals the edge
weights gradually, maintaining a matching and dual solution that satisfy an approximate
form of complementary slackness at the current scale. This reduces each scale to a
controlled relaxation of the exact problem, but it also requires managing blossoms
inherited from earlier scales. Much of the complexity of modern weighted matching
algorithms lies in handling these inherited blossoms efficiently~\cite{Gabow85scaling,GabowTarjan91}.

\paragraph{The Duan--Pettie--Su improvement.}
Duan, Pettie, and Su~\cite{dps} developed the \Liq{} framework, which gives the
current best exact combinatorial bound for integer-weight MWPM in general graphs. At a
high level, their algorithm organizes each scale around three main procedures:
\SBL{}, \FVR{}, and \Finalization{}. Our work preserves this
high-level organization, but replaces the \FVR{} routine by a
quantum-enhanced version.

\paragraph{Algebraic approaches.}
A different line of work uses fast matrix multiplication to obtain randomized exact
algorithms with running time polynomial in $W$ and $n^\omega$. These methods are most
competitive in dense regimes and are structurally quite different from the combinatorial
primal--dual framework considered here~\cite{CyganGabowSankowski15}.

\subsection{Quantum algorithms on matching problems}
\label{sec:qmatchreview}

There are two primary models for quantum matching algorithms.
 In the \emph{query model},   only the number of oracle calls made to access the graph are counted.
 Running time in the query model typically ignores
the cost of maintaining classical or quantum auxiliary data structures. 
By contrast, a \emph{time bound} is
stronger: it counts the total running time.
In addition, quantum graph algorithms may assume that the graph is given in
the form of an adjacency matrix or as a set of adjacency lists.
In the
adjacency-matrix model, a query asks whether a pair $(u,v)$ is an edge. 
In the
adjacency-list model, the algorithm may query the $i$th neighbor of a vertex. 
Our result
is a time bound in an adjacency-list model.
Table~\ref{tab:quantum-matching-prior} shows a summary of quantum algorithms for a variety of matching problems in various models. 
All of the matching problems for which quantum speedups have been obtained are for easier problems, in that they assume a bipartite graph or uniform weights on the edges. 
Our result provides the first quantum speedup for weighted matching in general graphs.

\begin{table}[h]
\centering
\small
\setlength{\tabcolsep}{5pt}
\renewcommand{\arraystretch}{1.15}
\begin{tabularx}{\linewidth}{|>{\raggedright\arraybackslash}p{0.29\linewidth}
                        |>{\raggedright\arraybackslash}p{0.21\linewidth}
                        |>{\raggedright\arraybackslash}X|}
\hline
\textbf{Problem} &
\textbf{Model / graph class} &
\textbf{Best-known quantum bound (selected)}\\
\hline
Maximum matching &
Oracle graph access; bipartite &
$O(n\sqrt{m+n}\log n)$ \cite{AS06} \\
\hline
Maximum matching &
Oracle graph access; general &
$O\!\left(n^2(\sqrt{m/n}+\log n)\log n\right)$ \cite{AS06} \\
\hline
Maximum matching &
Time bound; general &
$O(n\sqrt{m}\log^2 n)$ in the adjacency-list model \cite{Dor09} \\
\hline
Maximum bipartite matching &
Query bound; adjacency matrix &
$O(n^{7/4})$ \cite{LL15} \\
\hline
Maximum bipartite matching &
Query bound; adjacency list &
$O(n^{3/4}\sqrt{m+n})$ \cite{BT20} \\
\hline
Maximum matching &
Query bound; general &
$O(n^{7/4})$ matrix queries and $O(n^{3/4}\sqrt{m+n})$ list queries \cite{KW21} \\
\hline
Maximum bipartite matching &
Quantum edge-query model &
$O(n^{3/2}\log^2 n)$ \cite{Bli22} \\
\hline
Maximum-weight bipartite matching  &
Time bound; bipartite, integer weights &
$O(n\sqrt{m}W\log^2 n)$ \cite{Dor09} \\
\hline
Minimum-weight perfect matching &
Time bound; bipartite, weighted &
$O(n^{2.5}\log^3 n)$ in the matrix model \cite{Dor09,DHHM04} \\
\hline
$k$-matching &
Query bound; general &
$O(n\sqrt{k}+k^2)$ in the adjacency-matrix model \cite{TM24} \\
\hline
\end{tabularx}
\caption{Selected prior quantum results for matching problems. Here $n=|V|$, $m=|E|$,
and $W$ is the maximum integer edge weight when applicable.}
\label{tab:quantum-matching-prior}
\end{table}

Prior quantum matching algorithms obtain their speedups in settings where the search structure is simpler than in exact weighted matching on general graphs. For example, D\"orn~\cite{Dor09} uses Grover-style search to speed up the exploration steps in augmenting-path algorithms, obtaining a time bound for maximum cardinality matching in general graphs and for maximum-weight matching in bipartite graphs. Other works, like \cite{BT20}, study query complexity, where the goal is to minimize the number of graph-access queries rather than to give a full running-time bound including all data-structure updates.

The closest prior result to ours is due to D\"orn~\cite{Dor09} who gives quantum
speedups for two related matching problems: maximum cardinality matching in general
graphs, with running time $O(n\sqrt{m}\log^2 n)$ and maximum-weight matching in bipartite graphs, with running time $O(n\sqrt{m}\,W\log^2 n)$.
These two results capture two separate parts of the
difficulty addressed in this paper: the first allows general graphs but is unweighted, while
the second allows weights but only in the bipartite setting. Our problem combines both
features.

The weighted quantum algorithm of D\"orn does not work
inside a weight-scaling primal--dual framework. Its dependence on the largest edge weight
is linear in $W$, whereas scaling-based classical algorithms achieve only logarithmic
dependence on the weight range.  To the best of our knowledge, there were previously no known
quantum speedups for matching algorithms that use the scaling framework, allowing for a logarithmic dependence on $W$.

\section{Imported classical facts from Duan--Pettie--Su}
\label{app:dps}

\subsection{The \SBL{} call-count bound}

The proof of Lemma~\ref{le:detail} uses the following classical fact.

\begin{lemma}[\SBL{} search count, imported from~\cite{dps}]
\label{le:dps-detail}
Fix one scale and let \(B\) be a maximal old small blossom. During the
\SBL{} procedure, the number of calls to \textsc{PQSearch} restricted to
\(B\) is \(O(n(B))\), where \(n(B)\) is the number of vertices of \(B\).
\end{lemma}

\begin{proof}[Classical idea]
Each such call either matches at least two more vertices in \(B\), or strictly
enlarges the set of free vertices in \(B\) having maximum \(y\)-value. Since
both processes can happen only \(O(n(B))\) times, the total number of calls is
\(O(n(B))\).
\end{proof}

This is the only nontrivial classical fact needed in the proof of
Lemma~\ref{le:detail}; once it is imported, the rest of the argument is the
elementary estimate
\[
    O(n(B))\cdot \widetilde{O}(m(B)+n(B))
    =
    \widetilde{O}(n(B)\tau^2)
\]
for a small blossom \(B\).

\subsection{The free-vertex bound after one scale}

The proof of Lemma~\ref{le:clean} uses the following classical consequence of
the \Liq{} analysis.

\begin{lemma}[Free-vertex bound, imported from~\cite{dps}]
\label{le:dps-free}
After one scale of the \Liq{} algorithm, the number of free vertices remaining
after the \FVR{} procedure is \(O(n/\tau)\).
\end{lemma}

The proof in~\cite{dps} depends only on the maintained properties of the
framework and on the postcondition of the classical \SearchOne{} routine.
Since Theorem~\ref{th:qsearchone-correct} shows that \QSearchOne{} performs
the same five tasks as \SearchOne{}, the same bound applies to the
quantum-enhanced algorithm.

Summing over all \(O(\log W)\) scales yields the bound used in
Lemma~\ref{le:clean}:
\[
    O\!\left(\frac{n}{\tau}\log W\right)
\]
free vertices may appear in \Finalization{}.

\section{Matchings, blossoms, and dual variables}

\label{ap:lp}

Let $G=(V,E)$ be an undirected graph with edge-weight function
$w:E\to \mathbb{R}$. A \emph{matching} is a set of edges, no two of which share an
endpoint. A matching is \emph{perfect} if every vertex is incident to exactly one matched
edge.

In this paper we focus on \emph{maximum-weight perfect matching} (MWPM). Thus the goal is
to find a perfect matching $M\subseteq E$ that maximizes
\[
w(M) \;=\; \sum_{e\in M} w_e .
\]

\paragraph{The perfect-matching LP.}
For each edge $e\in E$, introduce a variable $x_e$ indicating whether $e$ is selected.
Let $\delta(v)$ denote the set of edges incident to $v$, and for $S\subseteq V$ let
$E(S)$ denote the set of edges with both endpoints in $S$.

The natural integer program for MWPM is
\begin{align*}
\max \ & \sum_{e\in E} w_e x_e \\
\text{s.t. } &
\sum_{e\in\delta(v)} x_e = 1 && \forall v\in V,\\
& x_e \in \{0,1\} && \forall e\in E.
\end{align*}
Dropping integrality gives the basic LP relaxation
\begin{align*}
\max \ & \sum_{e\in E} w_e x_e \\
\text{s.t. } &
\sum_{e\in\delta(v)} x_e = 1 && \forall v\in V,\\
& x_e \ge 0 && \forall e\in E.
\end{align*}
Every perfect matching induces a feasible solution to this relaxation, but in general
graphs the relaxation is not integral.

\paragraph{Odd sets and blossoms.}
The obstruction is parity. On an odd cycle, the assignment $x_e=\tfrac12$ for every edge
satisfies all degree constraints, although no integral perfect matching can do so. More
generally, for every odd set $S\subseteq V$ with $|S|\ge 3$, any matching can use at most
\[
\frac{|S|-1}{2}
\]
edges of $E(S)$. This yields the odd-set inequalities
\[
\sum_{e\in E(S)} x_e \le \frac{|S|-1}{2}
\qquad
\forall S\subseteq V \text{ odd},\ |S|\ge 3.
\]

Adding these inequalities gives Edmonds' exact LP formulation~\cite{Edmonds65poly}:
\begin{align}
\text{(P)} \qquad
\max \ & \sum_{e\in E} w_e x_e \notag\\
\text{s.t. } &
\sum_{e\in\delta(v)} x_e = 1 && \forall v\in V, \notag\\
& \sum_{e\in E(S)} x_e \le \frac{|S|-1}{2}
&& \forall S\subseteq V \text{ odd},\ |S|\ge 3, \notag\\
& x_e \ge 0 && \forall e\in E. \label{eq:primal-mwpm}
\end{align}

In blossom-based algorithms, the odd sets that become active during the search are
represented as \emph{blossoms}. Algorithmically, a blossom is an odd set that is treated as
a contracted object during alternating-tree search. The maintained blossoms form a
laminar family, which allows them to be handled hierarchically.

\paragraph{Dual variables.}
The dual of \eqref{eq:primal-mwpm} is
\begin{align}
\text{(D)} \qquad
\min \ & \sum_{v\in V} y(v)
\;+\;
\sum_{\substack{S\subseteq V\\ |S|\ge 3\ \text{odd}}}
\frac{|S|-1}{2}\,z(S)
\notag\\
\text{s.t. } &
y(u)+y(v)+\sum_{\substack{S\subseteq V \text{ odd}\\ e\in E(S)}} z(S)
\;\ge\; w_e
&& \forall e=(u,v)\in E, \notag\\
& z(S) \ge 0
&& \forall S\subseteq V \text{ odd},\ |S|\ge 3, \label{eq:dual-mwpm}
\end{align}
where the vertex variables $y(v)$ are free and the odd-set variables satisfy $z(S)\ge 0$.

In the primal--dual view, the variables $y$ and $z$ are the dual weights maintained by the
algorithm. For a current laminar blossom family $\Omega$, it is convenient to write
\[
yz(e)
\;:=\;
y(u)+y(v)+\sum_{B\in \Omega:\, e\subseteq B} z(B)
\qquad \text{for } e=(u,v)\in E,
\]
and define the \emph{slack} of an edge by
\[
S(e) \;:=\; yz(e)-w(e).
\]
An edge is \emph{tight} if $S(e)=0$.

\paragraph{Complementary slackness.}
A primal--dual pair $(x,(y,z))$ is optimal if it is feasible and satisfies complementary
slackness. For \eqref{eq:primal-mwpm} and \eqref{eq:dual-mwpm}, this means in particular:
\begin{itemize}
    \item if $x_e>0$, then
    \[
    y(u)+y(v)+\sum_{\substack{S\text{ odd}:\\ e\in E(S)}} z(S) = w_e;
    \]
    that is, every used edge is tight;

    \item if $z(S)>0$, then
    \[
    \sum_{e\in E(S)} x_e = \frac{|S|-1}{2};
    \]
    that is, every active odd set is tight.
\end{itemize}

This is the basic structural picture behind blossom algorithms: one searches for
augmenting structure among tight edges while maintaining dual feasibility, and the active
odd sets with positive $z$-value are exactly the blossom constraints that matter.

\paragraph{Bipartite specialization.}
In bipartite graphs, the odd-set inequalities are redundant. The dual then consists only
of vertex potentials $y$, and the blossom variables disappear. The general-graph setting is
harder precisely because odd sets must be tracked explicitly.

\section{Reduction remarks}

This appendix describes two simple reductions that justify the choice of problem
formulation in the main body.

\subsection{From maximum-weight matching to maximum-weight perfect matching}

We first show that maximum-weight matching without the perfect-matching requirement
reduces to maximum-weight perfect matching with only constant-factor overhead.

\begin{lemma}
\label{le:mm-to-mwpm}
Maximum-weight matching in a general graph reduces to maximum-weight perfect
matching with constant overhead.
\end{lemma}

\begin{proof}
Let $G=(V,E)$ be an edge-weighted graph. Construct a graph
\[
H=(V^{(1)}\cup V^{(2)},\,E^{(1)}\cup E^{(2)}\cup F)
\]
as follows. For each vertex $v\in V$, create two copies $v^{(1)}\in V^{(1)}$
and $v^{(2)}\in V^{(2)}$. For each edge $uv\in E$, add edges
$u^{(1)}v^{(1)}\in E^{(1)}$ and $u^{(2)}v^{(2)}\in E^{(2)}$, each with weight
$w(uv)$. Finally, for each $v\in V$, add a vertical edge
$v^{(1)}v^{(2)}\in F$ of weight $0$.

Then
\[
|V(H)|=2|V(G)|,
\qquad
|E(H)|=2|E(G)|+|V(G)|,
\]
so the construction enlarges the instance by only a constant factor.

Let $M$ be any matching in $G$. Define
\[
P(M)
=
\{u^{(1)}v^{(1)},u^{(2)}v^{(2)}: uv\in M\}
\;\cup\;
\{v^{(1)}v^{(2)}: v \text{ is unmatched in } M\}.
\]
Then $P(M)$ is a perfect matching in $H$, and
\[
w_H(P(M))=2w_G(M).
\]
Therefore
\[
\mathrm{OPT}_{PM}(H)\ge 2\,\mathrm{OPT}_{M}(G).
\]

Conversely, let $P$ be any perfect matching of $H$. Let $M_1$ be the projection
to $G$ of the edges of $P$ inside $V^{(1)}$, and define $M_2$ analogously from
the edges inside $V^{(2)}$. Since $P$ is a matching, both $M_1$ and $M_2$ are
matchings in $G$. Moreover, the vertical edge $v^{(1)}v^{(2)}$ appears in $P$
exactly when $v$ is unmatched inside both copies, so
\[
w_H(P)=w_G(M_1)+w_G(M_2)\le 2\,\mathrm{OPT}_{M}(G).
\]
Hence
\[
\mathrm{OPT}_{PM}(H)=2\,\mathrm{OPT}_{M}(G).
\]

Thus one call to a maximum-weight perfect matching algorithm on $H$, followed by
projection onto either copy, yields an optimal maximum-weight matching in $G$.
\end{proof}

\subsection{Maximum-weight versus minimum-weight perfect matching}

For perfect matchings, the maximum-weight and minimum-weight versions are
equivalent up to a trivial weight transformation.

\begin{lemma}
\label{le:max-min-perfect}
Maximum-weight perfect matching and minimum-weight perfect matching are
equivalent up to an affine transformation of the edge weights.
\end{lemma}

\begin{proof}
Let $G=(V,E)$ be a graph and let $w:E\to \mathbb{R}$ be edge weights. Choose a
constant $C>\max_{e\in E} w(e)$, and define
\[
w'(e)=C-w(e)
\qquad \forall e\in E.
\]

Every perfect matching of $G$ contains exactly $|V|/2$ edges. Therefore, for any
perfect matching $M$,
\[
w'(M)
=
\sum_{e\in M}(C-w(e))
=
\frac{|V|}{2}\,C - w(M).
\]
Since the first term is the same for every perfect matching, minimizing $w'(M)$
is equivalent to maximizing $w(M)$. Thus an algorithm for minimum-weight perfect
matching immediately yields an algorithm for maximum-weight perfect matching, and
vice versa.
\end{proof}

\section{Contracted and original graph paths}
\label{sec:contracted}

Let \(H=G/\Omega\) denote the graph obtained by contracting each current root
blossom to one supervertex. For a blossom \(B\), let \(b(B)\) denote its base.
We use the standard blossom invariant that every nontrivial blossom is stored
with an odd alternating cycle of child blossoms, with one distinguished base
child. Recursively, the base of a blossom is the base of its base child. If a
blossom has a matched edge crossing its boundary, then that edge is incident to
its base; if it has no matched edge crossing its boundary, then the blossom is
free.

\begin{lemma}[Internal expansion inside a blossom]
\label{le:internal-expansion}
Let \(B\) be a blossom and let \(x\in B\). There is an alternating path inside
\(B\) from \(b(B)\) to \(x\), determined by the stored defining cycles of the
blossoms contained in \(B\). Such a path can be constructed recursively in time
linear in the number of edges output.
\end{lemma}

\begin{proof}
The proof is by induction on the depth of \(B\) in the blossom forest. If
\(B\) is trivial, then \(B=\{x\}\), and the claim is immediate.

Suppose \(B\) is nontrivial, with children
\[
    C_0,C_1,\ldots,C_{2k},
\]
where \(C_0\) is the base child. Let \(C_i\) be the child containing \(x\).
The children lie on an odd alternating cycle. Along this cycle, there is a
stored alternating route from the base child \(C_0\) to \(C_i\) with the parity
needed to connect the base of \(B\) to the desired endpoint in \(C_i\). For
each child blossom encountered on this route, we recursively expand the needed
internal segment using the induction hypothesis.

Concatenating the expanded child segments with the defining-cycle edges gives
an alternating path inside \(B\) from \(b(B)\) to \(x\). Since the construction
outputs only the pieces of blossom structure used by the path, the running time
is linear in the number of edges output.
\end{proof}

\begin{lemma}[Base-unavoidability]
\label{le:base-unavoidable}
Let \(B\) be a current root blossom. Any augmenting path in \(G\) that uses
vertices of \(B\) and is compatible with the current blossom contraction
contains the base \(b(B)\). In particular, the lift of any augmenting path of
blossoms in \(H\) contains the base of every root blossom appearing on the
contracted path.
\end{lemma}

\begin{proof}
If \(B\) is an endpoint blossom of an augmenting path of blossoms in \(H\), then
\(B\) is free. By the blossom invariant, the only possible unmatched vertex in
a free blossom is its base. Hence the lifted augmenting path must start or end
at \(b(B)\).

Now suppose \(B\) is an internal blossom on an augmenting path. The contracted
path enters and leaves \(B\) using one matched and one unmatched boundary edge.
By the blossom invariant, the only matched edge crossing the boundary of \(B\),
if such an edge exists, is incident to \(b(B)\). Therefore any alternating
traversal through \(B\) must pass through \(b(B)\). The internal portion from
\(b(B)\) to the other boundary endpoint is supplied by
Lemma~\ref{le:internal-expansion}. Thus every lifted traversal of \(B\)
contains \(b(B)\).
\end{proof}

\begin{proof}[Proof of Fact~\ref{fact:aug1}]
Let
\[
    B_1,e_1,B_2,e_2,\ldots,e_{r-1},B_r
\]
be an augmenting path of blossoms in \(H=G/\Omega\). Thus \(B_1\) and \(B_r\)
are free root blossoms, each \(e_i\) is an original edge with one endpoint in
\(B_i\) and one endpoint in \(B_{i+1}\), and the edges
\(e_1,\ldots,e_{r-1}\) alternate between unmatched and matched.

We lift this path by expanding its portion inside each blossom. In the first
blossom \(B_1\), start at the base \(b(B_1)\) and use
Lemma~\ref{le:internal-expansion} to reach the endpoint of \(e_1\) inside
\(B_1\). In the last blossom \(B_r\), use
Lemma~\ref{le:internal-expansion} to connect the endpoint of \(e_{r-1}\) inside
\(B_r\) to the base \(b(B_r)\).

For an internal blossom \(B_i\), the path uses one edge entering \(B_i\) and
one edge leaving \(B_i\). Since the contracted path is alternating, exactly one
of these two boundary edges is matched. By the blossom invariant, this matched
boundary edge is incident to \(b(B_i)\). Therefore the internal part of the
lifted path inside \(B_i\) is obtained by connecting \(b(B_i)\) to the endpoint
of the other boundary edge, again using Lemma~\ref{le:internal-expansion}.

Concatenating these internal expansions with the boundary edges
\(e_1,\ldots,e_{r-1}\) gives an alternating path in the original graph \(G\).
Because \(B_1\) and \(B_r\) are free, the path starts at \(b(B_1)\) and ends at
\(b(B_r)\), both of which are unmatched vertices. Hence the lifted path is an
augmenting path in \(G\). The construction outputs only the internal blossom
segments used by the lifted path, so its running time is linear in the length
of the lifted path.

It remains to check that augmenting along the lifted path preserves the blossom
structure invariants. Augmentation flips the matching status of the edges on an
alternating path. Inside every blossom crossed by the lifted path, this flip
only changes which child acts as the base child; the odd defining cycle and the
laminar set of blossom vertices remain unchanged. Recursing through nested
blossoms shows that every blossom still has a valid base, a valid alternating
defining cycle, and at most one matched edge crossing its boundary. Therefore
augmenting along the lifted path preserves the blossom-structure invariants.
\end{proof}

\begin{proof}[Proof of Fact~\ref{fact:aug2}]
Let \(\mathcal P\) be a maximal set of blossom-disjoint augmenting paths in
\(H=G/\Omega\), and lift every path in \(\mathcal P\) to \(G\) using
Fact~\ref{fact:aug1}.

First, the lifted paths are vertex-disjoint. Indeed, two different contracted
paths in \(\mathcal P\) use disjoint sets of root blossoms. The lift of a
contracted path uses only original vertices contained in the root blossoms on
that path. Since distinct root blossoms are disjoint as vertex sets, the lifted
paths are vertex-disjoint.

\end{proof}

\begin{proof}[Proof of Fact~\ref{fact:blossommatch}]
We prove the claim by induction on the blossom tree.

If \(B\) is a trivial blossom, then \(B\) consists of a single original vertex.
If this vertex is unmatched, then it is the base of \(B\) and \(B\) is free.
If it is matched, then \(B\) contains no unmatched vertex and is matched.
Thus the claim holds for trivial blossoms.

Now suppose \(B\) is nontrivial, with children
\[
    C_0,C_1,\ldots,C_{2k},
\]
where \(C_0\) is the base child. The defining cycle of \(B\) contains \(k\)
matched edges that pair up all children except the base child \(C_0\). By the
induction hypothesis, each child \(C_i\) contains at most one vertex that could
be unmatched inside that child, namely its base. Whenever a child is incident
to one of the matched defining-cycle edges of \(B\), the base of that child is
matched by this cycle edge. Hence none of the non-base children
\(C_1,\ldots,C_{2k}\) contains an unmatched vertex in the full matching.

Therefore the only possible unmatched vertex inside \(B\) is the base of the
base child \(C_0\), which is exactly \(b(B)\). If \(B\) has no matched edge
crossing its boundary, then \(b(B)\) remains unmatched, and \(B\) is free. If
\(B\) has a matched edge crossing its boundary, then by the blossom invariant
this edge is incident to \(b(B)\). In that case \(b(B)\) is also matched, and
there is no unmatched vertex inside \(B\).

Thus \(B\) contains at most one unmatched vertex. If such a vertex exists, it
is \(b(B)\) and \(B\) is free. If no such vertex exists, then \(B\) is matched.
\end{proof}

\end{document}